\documentclass[journal]{IEEEtran}
\ifCLASSINFOpdf
\else
\fi
\usepackage{amsthm}
\usepackage{amssymb}
\usepackage{amsmath}
\usepackage{graphicx}
\usepackage{subfigure}
\usepackage{caption}
\usepackage{pmat}
\usepackage{float}
\usepackage{lipsum}
\usepackage{enumitem}
\usepackage{epstopdf}
\usepackage{multirow}
\usepackage{tabularx}
\usepackage{lipsum}
\usepackage{multicol}
\usepackage{amsfonts}
\usepackage[english]{babel}
\usepackage[autostyle]{csquotes}
\usepackage{array}
\usepackage{dcolumn}
\theoremstyle{remark}
\usepackage{booktabs}

\usepackage{arydshln}

\usepackage{algorithm}
\usepackage[end]{algpseudocode}
\usepackage{setspace}
\usepackage{nicefrac}     
\usepackage{microtype}
\usepackage[bookmarks=false]{hyperref}
\usepackage{hyperref}
\usepackage{xcolor}
\usepackage{mathtools, cuted}
\usepackage{flushend}
\usepackage{xfrac}
\begin{document}
	 \thispagestyle{empty} 
	\onecolumn
	\begin{center}
This paper has been accepted for publication in \textit{IEEE TRANSACTIONS ON   SYSTEM, MAN, AND CYBERNETICS: SYSTEMS}\\\vspace{1cm}
DOI: {\color{blue}10.1109/TSMC.2020.3012507}\\
IEEE Xplore: {\color{blue}\url{https://ieeexplore.ieee.org/document/9166758}}
\end{center}
\vspace{1cm}
@2020 IEEE. Personal use of this material is permitted. Permission from IEEE must be obtained for all other uses, in any
current or future media, including reprinting /republishing this material for advertising or promotional purposes, creating new
collective works, for resale or redistribution to servers or lists, or reuse of any copyrighted component of this work in other works.

\newpage
\title{Robust   Simultaneously  Stabilizing  Decoupling Output Feedback Controllers for   Unstable Adversely Coupled   Nano Air Vehicles}

\author{Jinraj~V.~Pushpangathan,~
        Harikumar~Kandath,~\IEEEmembership{Member,~IEEE,},
          Suresh~Sundaram,~\IEEEmembership{Senior Member,~IEEE,},
        and~Narasimhan~Sundararajan,~\IEEEmembership{Life Fellow,~IEEE}
\thanks{Research fellow @ Department
of Aerospace Engineering, Indian Institute of Science, Bangalore-560012,
India, e-mail: (jinrajaero@gmail.com).}
\thanks{Assistant Professor @ Robotics Research Center, International Institute of Information Technology, Hyderabad-32, India, e-mail:(harikumar.k@iiit.ac.in)}
\thanks{Associate professor @ Department
of Aerospace Engineering, Indian Institute of Science, Bangalore- 560012, India, e-mail: (vssuresh@iisc.ac.in).}
\thanks{Professor (Retd.) @ School of Electrical and Electronics Engineering, Nanyang Technological University, Singapore, e-mail: (ensundara@ntu.edu.sg).}
}
\markboth{IEEE TRANSACTIONS ON SYSTEM, MAN, AND CYBERNETICS: SYSTEMS,~Vol.~XX, No.~XX, XXXX~2020}%
{Shell \MakeLowercase{\textit{et al.}}: Bare Demo of IEEEtran.cls for IEEE Journals}
\twocolumn

\maketitle
\setcounter{page}{1}

\begin{abstract}
The plants of nano air vehicles (NAVs) are generally unstable, adversely coupled, and uncertain. Besides, the autopilot hardware of a NAV has limited sensing and computational capabilities. Hence, these vehicles need a single controller referred to as Robust Simultaneously Stabilizing Decoupling (RSSD)  output feedback controller that achieves simultaneous stabilization,  desired decoupling, robustness, and performance for a finite set of unstable multi-input-multi-output adversely coupled uncertain plants. To synthesize a RSSD  output feedback controller, 
 a new method that is based on a central plant  is proposed in this paper. Given a finite set of plants for simultaneous stabilization, we considered a plant  in this set that has the smallest maximum $v-$gap metric as the central plant. 
Following this,  the sufficient condition for the existence of a simultaneous stabilizing controller associated with such a plant  is described.   The decoupling feature is then appended to this controller using the properties of the eigenstructure assignment method.
 Afterward, the sufficient conditions for the existence of a RSSD  output feedback controller are obtained. Using these sufficient conditions, a new optimization problem for the synthesis of a RSSD  output feedback controller is formulated.  To solve this optimization problem, a new genetic algorithm based offline iterative algorithm is developed.  The effectiveness of this iterative algorithm is then demonstrated by generating a RSSD controller for a fixed-wing nano air vehicle. The performance of this controller is validated through numerical and hardware-in-the-loop simulations.

\end{abstract}
\IEEEpeerreviewmaketitle
\section{Introduction}

Nano air vehicles  are extremely small air vehicles that are widely used for intelligence, battlefield surveillance and reconnaissance, and disaster assessment missions.  NAVs have severe dimensional and weight constraints as their overall dimension and weight need to be lower than 75~mm and 20~g, respectively \cite{petericca}. Based on flight modes, they are classified as fixed-wing, rotary-wing, and flapping-wing NAVs. Figure \ref{fig:foo}  shows   a typical fixed-wing NAV that   weighs 19.4~g and has an overall dimension of 75~mm. Generally, a fixed-wing NAV requires flight controllers to accomplish a mission. Note that in this paper,  \textit{plant}  represents \textit{plant model}.
In most cases, one designs these flight controllers based on  linear time-invariant (LTI) plants that are obtained by linearizing the nonlinear plants of a NAV about   the nominal trajectory associated with different steady-state flight conditions. Thereafter, these controllers are scheduled (gain scheduling) based on the  altitude and Mach number (scheduling variables) to control the nonlinear plant. For example, in \cite{kami}, a gain scheduled controller for trajectory tracking is synthesized for a fixed-wing unmanned aerial vehicle (UAV). Because of the weight and dimensional constraints, the  autopilot hardware of a NAV has severe resource constraints as it lacks lightweight sensors to measure  state variables like an angle-of-attack, airspeed, and sideslip angle. 
Also, NAV's autopilot hardware does not possess sufficient computational and memory powers.  Figure \ref{fig:autopilot} shows a typical autopilot hardware of the 75~mm wingspan fixed-wing NAV \cite{jin}. This autopilot weighs 1.8~g and  has no sensors to measure airspeed, angle-of-attack, and  sideslip angle.
\begin{figure}[H]
\centering
\subfigure[Fixed-wing NAV \label{fig:foo}]{\includegraphics[width=1.5in, height=0.9in]{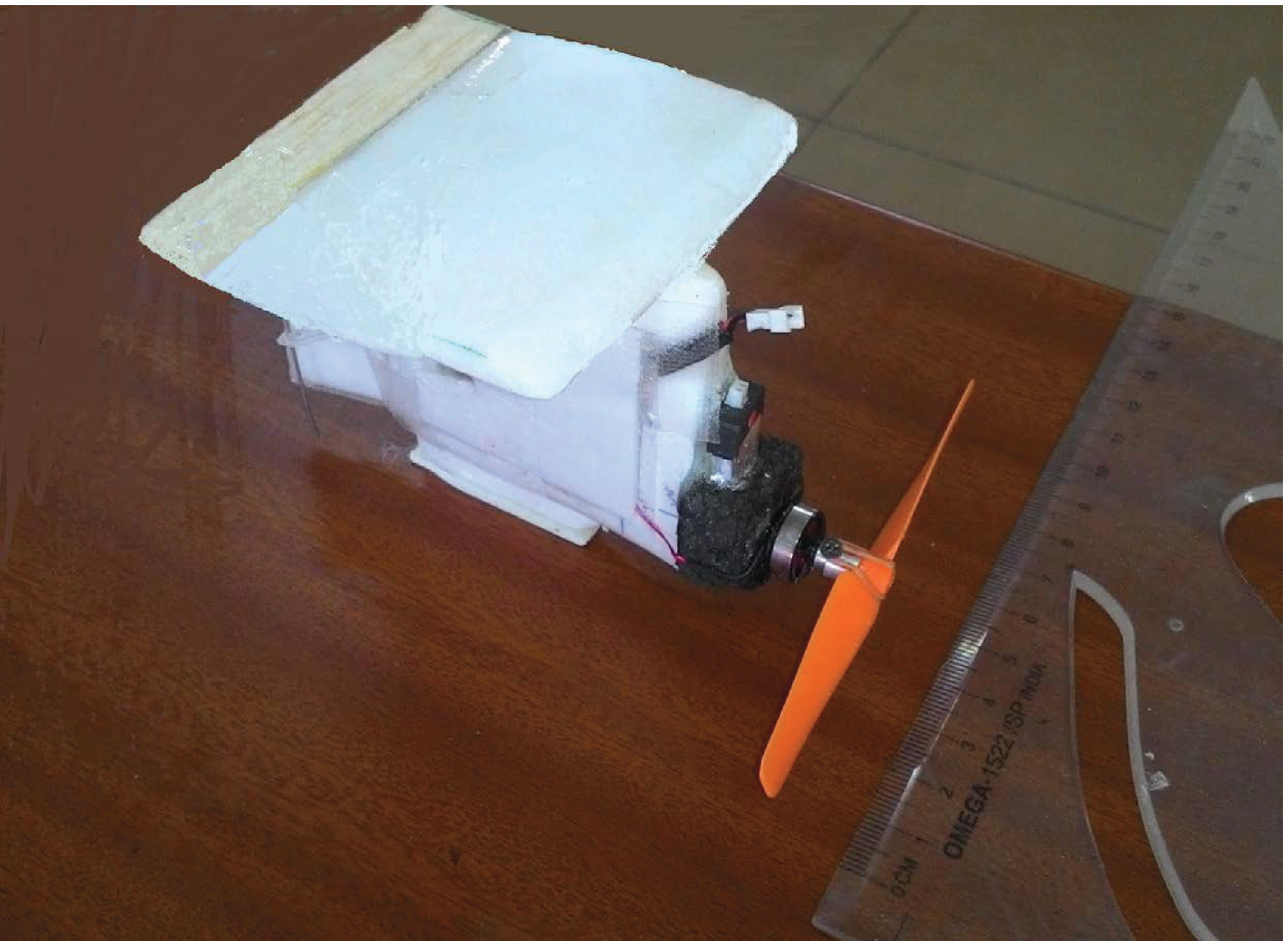}}
\subfigure[Autopilot hardware  \label{fig:autopilot}]{\includegraphics[width=1.5in, height=0.9in]{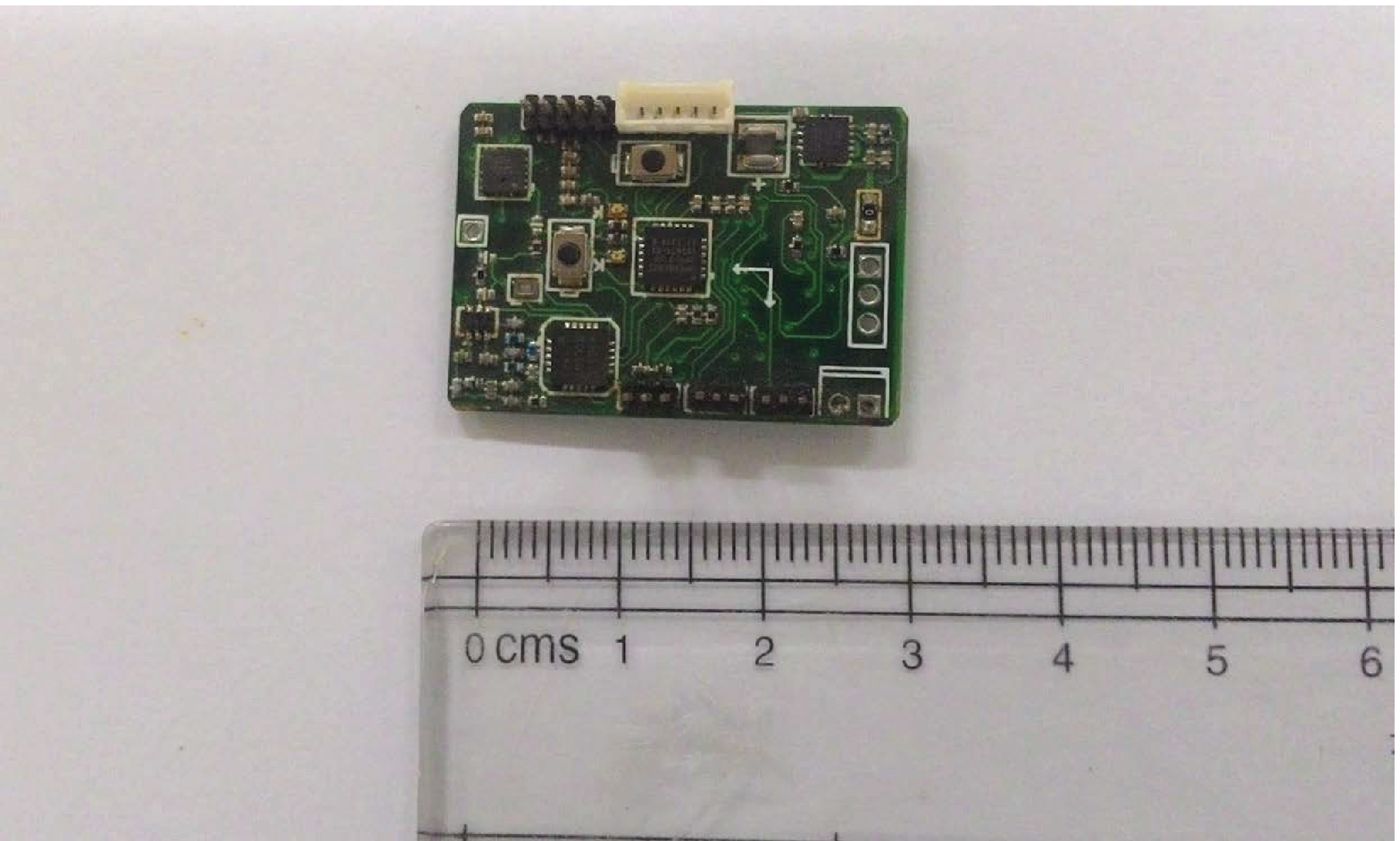}}
 \caption{75~mm wingspan NAV and  autopilot hardware}
\end{figure} 
\noindent Besides, this autopilot hardware has only 1~MB of memory capacity. Due to the insufficient computational and memory powers, it is hard to implement  the widely used estimation algorithm, namely, the Extended Kalman Filter  (EKF)     on  NAV's autopilot hardware. As airspeed cannot be measured or estimated, the Mach number will not be available for the scheduling of  controllers. Hence,  gain scheduling becomes inappropriate for the NAV. Generally, the controller for an aircraft is designed using  full state feedback.  For instance, in \cite{avati}-\cite{avati1} and  \cite{wei},  a full state feedback  controller is designed for a quadcopter and a micro aerial vehicle, respectively. However, the flight controller design  using full state feedback becomes ineffective for the NAV as all the state variables cannot be measured or estimated. The multi-input-multi-output (MIMO)  LTI plants of a fixed-wing NAV are generally unstable because of significant adverse coupling (mainly due to the notable gyroscopic coupling and counter torque of the propulsion unit), dimensional constraints on the stabilizers, aerodynamic effects,  and due to the limitation in keeping center-of-gravity (CG) at the desired location  \cite{jin}, \cite{harip}. 
The significant adverse coupling also makes the successive loop closure autopilot design technique ill-suited for a fixed-wing NAV. Further, the autopilot design using the distributed control system approach  \cite{lan} becomes unsuitable due to adverse coupling and unavailability of scheduling variables. Besides this, the presence of coupling in the dynamics  makes various existing well-proven waypoint following algorithms unsuitable for a fixed-wing NAV as   these algorithms are developed by considering a decoupled longitudinal and lateral dynamic plants. Hence, for using these  algorithms, the flight controller needs to have a mode decoupling capability. Furthermore, the plants of a fixed-wing NAV  have significant  uncertainties mainly due to the unsteady flow, wind effect, unmodeled dynamics, and inaccurate measurements of forces and moments. The  considerable uncertainties in  the plant  cause notable errors in the estimation of  state variables. Also, the unmodeled dynamic uncertainty makes a linear quadratic regulator (LQR) design unsuitable for a fixed-wing NAV.  Besides gain scheduling, the controller design techniques employed to design the flight controller for a nonlinear plant  of small aircraft are feedback linearization \cite{wenya},  backstepping with sliding mode control \cite{espi}, and model reference adaptive control \cite{zach}. However,  it is difficult to implement these controllers in a NAV due to the absence of high-bandwidth actuators and lightweight autopilot hardware with sufficient computational power and suitable sensors. In general,  MIMO LTI plants of a fixed-wing NAV are unstable, adversely coupled, and uncertain    \cite{jinthesis}. 
These general characteristics of the LTI plants and the resource constraints of autopilot hardware     suggest    that a NAV requires  a single computationally simple multivariable    output feedback  controller that provides the simultaneous   stabilization, desired decoupling, robustness, and performances to all the plants. \par

The determination of a single controller that stabilizes a finite number of plants is referred to as the  simultaneous stabilization (SS) problem. A simultaneously stabilizing  controller is suitable for a NAV as it needs a single controller that stabilizes a finite number of  LTI plants. The SS problem was first studied in  \cite{Saek} and \cite{vidya}. There is  no tractable solution to the SS problem of more than two plants  due to its NP-hard nature \cite{Gever}-\cite{onur}. Hence, for these cases, the SS problem is solved through numerical means. In \cite{CAO1},  an iterative algorithm based on linear matrix inequalities is developed to solve the SS problem of a finite set of strictly proper MIMO plants using static full state feedback and output feedback. A decomposition strategy  to solve the SS problem was proposed in  \cite{perez}. Here,  a bi-level design optimization structure is selected in which the design of the single controllers for each  plant  is carried out at the bottom level, whereas at the top-level optimization generates a single controller which can approximate all those individual controllers.
 Further solutions to the SS problem employing optimization techniques can be found in \cite{wu}-\cite{pas}.
The SS problem can also be solved by first deriving a sufficiency condition associated  with a central plant  and then synthesizing a controller that  satisfies this sufficiency condition using robust stabilization techniques. 
This method is followed in \cite{saif}. Here, the central plant is obtained by solving a 2-block optimization problem. 
 In \cite{jinjgcd},   a central plant  which is the best \textit{approximation} (\textit{closest} plant) to all the plants that require SS is identified. Then, the sufficient condition for the existence of a simultaneous stabilizing controller associated with this central plant  is obtained using the robust stabilization condition of a coprime factorized plant  and $v-$gap metric.  Following this, a stabilizing controller that satisfies this condition is synthesized.  It is to be noted that the aforementioned methods do not generate a single controller that concurrently achieves simultaneous stabilization, robustness, performance, and decoupling. \par

The problem of finding a single  output feedback controller that provides simultaneous stabilization, desired decoupling, robustness (against parametric and unmodeled dynamics uncertainties), and performances to a finite set of plants is considered as a Robust Simultaneous Stabilization Decoupling  (RSSD) problem. In this paper,   a RSSD problem is first formulated for a finite set of  MIMO LTI unstable adversely coupled uncertain plants of a fixed-wing NAV and then the same is solved in two steps. In the first step, the sufficient conditions for the existence of a robust simultaneous stabilizing output feedback controller associated with the central plant  are obtained using robust stabilization theory and the properties of the $v$-gap metric.
Next, an  output feedback decoupling controller that solves these conditions is synthesized using the eigenstructure assignment technique. 
The main contributions of this paper are 
\begin{enumerate}
\item As per the author's knowledge, this is the first time an output feedback controller design method that provides simultaneous stabilization, desired decoupling, robust stability, and performances is proposed.  Besides, the proposed method is generic as it can be applied to a finite set of MIMO plants, minimum/non-minimum phase plants, tall or fat plants, stable or unstable plants.

\item The sufficient conditions for the existence of a RSSD  output feedback controller are developed using the robust stabilization condition of the right coprime factorized plant,  the properties of $v$-gap metric,    and the eigenstructure assignment algorithm for output feedback.

\item Utilizing the sufficient conditions, a new  genetic algorithm based offline iterative algorithm called   NN-RSSD (non-convex-non-smooth-RSSD) is developed to solve a RSSD problem.

\item The effectiveness of NN-RSSD algorithm is demonstrated by generating a RSSD  output feedback controller for eight MIMO LTI unstable adversely coupled uncertain plants of the fixed-wing  NAV mentioned in \cite{jin}. The performance of this designed controller is then demonstrated first using   six-degree-of-freedom simulations. The implementation aspects, especially of the autopilot implementation, have been demonstrated using hardware-in-the-loop simulations of the closed-loop (CL) nonlinear plants of the NAV. \end{enumerate}

The  paper is organized as follows. 
In Section \ref{PRE}, preliminaries are given.
 The synthesizing of a RSSD  output feedback controller for a fixed-wing NAV  is explained in Section \ref{NRSSD}. 
 In Section \ref{DEex}, the design and performance evaluation of a RSSD  output feedback controller for the fixed-wing NAV is presented. Finally, in Section \ref{concu}, the conclusions  are  given.
\section{ PRELIMINARIES}\label{PRE} 
\newtheorem{defn}{Definition}[section]
\begin{defn}
\textit{Generalized stability margin ($b_{\mathbf{P},\mathbf{K}}$ $\in$ $[0, 1]$):}~~{\normalfont Given a plant, $\mathbf{P}(s)$   and its controller, $\mathbf{K}$. If the CL plant , $[\mathbf{P}(s), \mathbf{K}]$ is internally stable, then \cite{steele}
\begin{equation}
b_{\mathbf{P},\mathbf{K}}=\sfrac{1}
{\bigg|\bigg|\left[ \begin{array}{c} {\mathbf{P}(s)}  \\{\mathbf{I}}  \end{array} \right ]\left(\begin{array}{c} {\mathbf{I-KP}(s)}  \end{array} \right )^{-1}\left[ \begin{array}{cc} {\mathbf{-I}} & {\mathbf{K}}  \end{array} \right ]\bigg|\bigg|_\infty}
\label{eq:RCF1_1_1}
\end{equation}
otherwise, $b_{\mathbf{P},\mathbf{K}}$=0.
}

\end{defn}
\begin{defn}
\textit{$v-$gap metric:}~~{\normalfont Given two plants, $\mathbf{P}_1(s)$ and $\mathbf{P}_2(s)$. Then, $v-$gap metric, $\delta_v(\mathbf{P}_1(j\omega),\mathbf{P}_2(j\omega)) \in [0,1] $ of  $\mathbf{P}_1(s)$ and $\mathbf{P}_2(s)$ is defined as \cite{82}
\begin{align}
\intertext{if det(\textbf{I}+$\mathbf{P}_1(j\omega)\mathbf{P}_2(j\omega))~ \neq 0 ~\forall ~\omega~$  and
wno det(\textbf{I}+$\mathbf{P}_2(j\omega)\mathbf{P}_1(j\omega))+\eta(\mathbf{P}_1(j\omega))-\eta(\mathbf{P}_2(j\omega))$ $-\eta_0(\mathbf{P}_2(j\omega))$=0, then}
\delta_v(\mathbf{P}_1(j\omega),\mathbf{P}_2(j\omega))&=
\parallel \Psi(\mathbf{P}_1(j\omega),\mathbf{P}_2(j\omega)) \parallel_\infty \nonumber\\
\intertext{otherwise }  
 \delta_v(\mathbf{P}_1(j\omega),\mathbf{P}_2(j\omega))&=1   
\label{vgmetrc}
\end{align}
where $ \Psi(\mathbf{P}_1(j\omega),\mathbf{P}_2(j\omega))$ is $((\mathbf{I}+\mathbf{P}_2(j\omega)\mathbf{P}^*_2(j\omega))^{-0.5}$ $(\mathbf{P}_1(j\omega)-\mathbf{P}_2(j\omega))(\mathbf{I}+\mathbf{P}_1(j\omega)\mathbf{P}^*_1(j\omega))^{-0.5})$.
} 
\end{defn}
\begin{defn}
\textit{Central plant:}~~{\normalfont Given a finite set,  $\mathcal{P}$. Then, the central plant  belongs to $\mathcal{P}$   has the smallest maximum $v-$gap metric among the maximum $v-$gap metrics of all the plants in  $\mathcal{P}$  \cite{jinjgcd}.  In this paper, the central plant  of a set is denoted by subscript `$cp$'.}
\end{defn}

\section{Design  of a Robust Simultaneous Stabilizing Decoupling   Controller for a Fixed-Wing NAV}\label{NRSSD}
A fixed-wing  NAV  is developed to operate in a remotely controlled mode for surveillance without getting detected by radar.  Therefore,  the NAV needs to be stabilized in all operating conditions and hence requires a flight controller. This section explains the synthesizing of a RSSD  output feedback controller for a finite set of LTI MIMO unstable adversely coupled uncertain plants of a fixed-wing NAV. 

\subsection{Fixed-Wing NAV Dynamics and the Problem Statement}\label{P_S}
Typically, an aircraft dynamics is represented by a set of continuous nonlinear first-order differential equations given by
\begin{equation}
{\dot{\mathbf{X}}}=\mathbf{f(X, U)}
\label{nl}
\end{equation}
where   $\mathbf{X}$ is  the state vector, $\mathbf{U}$ is the input vector, and $\mathbf{f(\cdot, \cdot)}$ is a nonlinear vector function  describing  the dynamics. The control system design and stability analysis of an aircraft using these  nonlinear equations are complex. Instead,  an aircraft's LTI  plants are used for   the control system design and stability analysis. These LTI plants are obtained by linearizing the  nonlinear equations given in (\ref{nl}) about a nominal trajectory associated with various steady-state flight conditions.  The above mentioned  linear plants capture the approximate behavior of the nonlinear plant  in the neighborhood of a trim/equilibrium point belonging to the nominal trajectory. Further, the linear plant of an aircraft is decoupled into two, namely, linear  longitudinal and lateral models, when the gyroscopic coupling, aerodynamic cross-coupling, and inertial coupling in the dynamics of an aircraft are~negligibly small. These linear longitudinal and lateral models in state-space forms  are given by
\begin{align}
\dot{\mathbf{x}}_{{Lo}_f}&=A_{{Lo}_f}\mathbf{x}_{{Lo}_f}+B_{{Lo}_f}\mathbf{\delta}_{{Lo}_f}\label{adcp1}\\
\dot{\mathbf{x}}_{{La}_f}&=A_{{La}_f}\mathbf{x}_{{La}_f}+B_{{La}_f}\mathbf{\delta}_{{La}_f}\label{adcp2}
\end{align}
where $\mathbf{x}_{{Lo}_f} \in \mathbb{R}^{\hat{a}}$, $\mathbf{x}_{{La}_f} \in \mathbb{R}^{\hat{b}} $,  $\mathbf{\delta}_{{Lo}_f} \in \mathbb{R}^{\hat{c}}$, and $\mathbf{\delta}_{{La}_f} \in \mathbb{R}^{\hat{d}}$
are longitudinal state vector,  lateral state vector,  longitudinal  control vector, and  lateral control vector, respectively. In (\ref{adcp1})-(\ref{adcp2}), $A_{{Lo}_f} \in \mathbb{R}^{\hat{a} \times \hat{a}}$, $A_{{La}_f} \in \mathbb{R}^{\hat{b} \times \hat{b}}$, $B_{{Lo}_f} \in \mathbb{R}^{\hat{a} \times \hat{c}}$, and $B_{{La}_f} \in \mathbb{R}^{\hat{b} \times \hat{d}}$ represent system and control matrices of the longitudinal and lateral state-space models, respectively.
The dynamics of a fixed-wing NAV has significant cross-coupling  due to the inertial, gyroscopic, and aerodynamic effects \cite{jin}. Because of this cross-coupling, the longitudinal state variables have a significant influence from the lateral modes. Similarly, the lateral state variables have  significant influence from the longitudinal modes. Due to the significant cross-coupling,  the linear longitudinal and lateral models given in (\ref{adcp1}) and (\ref{adcp2}), respectively, are not suitable. The suitable linear model for a fixed-wing NAV is the linear coupled model which is given as
\begin{align}
\dot{\mathbf{x}}_f=A_{f}\mathbf{x}_f+B_{f}\mathbf{\delta}_{f}
\label{adcp3}
\end{align}
where $\mathbf{x}_f \in \mathbb{R}^{\hat{n}}=[\mathbf{x}_{Lo}|\mathbf{x}_{La}]^T$, $A_{f} \in \mathbb{R}^{(\hat{n} \times \hat{n})}=\left[\begin{array}{c;{2pt/2pt}r}A_{{Lo}_f}&A_{{Lo}_f}^{La}\\\hdashline[2pt/2pt]A_{{La}_f}^{Lo}&A_{{La}_f} \end{array}\right]$, $B_{f} \in\mathbb{R}^{(\hat{n} \times \hat{m})}=\left[\begin{array}{c;{2pt/2pt}r}B_{{Lo}_f}&B_{{Lo}_f}^{La}\\\hdashline[2pt/2pt]B_{{La}_f}^{Lo}&B_{{La}_f} \end{array}\right]$,   $\mathbf{\delta}_{f} \in \mathbb{R}^{\hat{m}} =[\mathbf{\delta}_{Lo}|\mathbf{\delta}_{La}]^T$, $\hat{n}=(\hat{a}+\hat{b})$, and $\hat{m}=(\hat{c}+\hat{d})$. Here, $A_{{Lo}_f}^{La}  \in \mathbb{R}^{\hat{a} \times \hat{b}}$, $A_{{La}_f}^{Lo}  \in \mathbb{R}^{\hat{b} \times \hat{a}}$,  $B_{{Lo}_f}^{La}  \in \mathbb{R}^{\hat{a} \times \hat{d}}$, and $B_{{La}_f}^{Lo}  \in \mathbb{R}^{\hat{b} \times \hat{c}}$ are the longitudinal coupling block of $A_f$, lateral coupling block of $A_f$, longitudinal coupling block of $B_f$, and lateral coupling block of $B_f$, respectively.
Along with cross-coupling, if  any $\lambda_{i_f}(\left[\begin{array}{c;{2pt/2pt}r}A_{{Lo}_f}&0\\\hdashline[2pt/2pt] 0&A_{{La}_f} \end{array}\right]) \in \mathcal{C}_-$ migrates towards $\mathcal{C}_+$ due to the presence of $A_{{Lo}_f}^{La}$ and $A_{{La}_f}^{Lo}$ in $A_f$, then  linear dynamics is  adversely coupled. Generally,  the linear dynamics   of a fixed-wing NAV is adversely coupled. One such  example is the  75~mm wingspan fixed-wing NAV mentioned in  \cite{jin}.
In \cite{jin}, it is shown that  $\lambda_{8_f}(\left[\begin{array}{c;{2pt/2pt}r}A_{{Lo}_f}&0\\\hdashline[2pt/2pt] 0&A_{{La}_f} \end{array}\right])=-0.112$ becomes 0.951 due to the presence of $A_{{Lo}_f}^{La}$ and $A_{{La}_f}^{Lo}$ in $A_f$. 
 Apart from the adverse coupling, all the linear plants associated with various steady-state flight conditions defined in the flight envelope of this NAV are unstable with a different number of unstable poles \cite{jinthesis}.  Additionally, the propeller effects, wind effects, sensors dynamics are not considered while developing the dynamic model of 75~mm wingspan fixed-wing NAV. Hence, the plants of this NAV have significant parametric and unmodeled dynamic uncertainties. Also, the autopilot hardware of the fixed-wing NAV  shown in Fig. \ref{fig:autopilot} has resource constraints. Therefore, to accomplish a mission, this NAV requires a RSSD  output feedback flight controller as the gain scheduling, full state feedback controller, successive-loop-closure, and LQR do not suit.\par
Now, to define the problem statement, let $\mathbf{P}_f(s) \in \mathcal{RL_\infty}^{\hat{r} \times \hat{m}}$ be the stabilizable LTI MIMO unstable   adversely coupled  uncertain plant  of a NAV. Here,  $\mathcal{RL_\infty}^{\hat{r} \times \hat{m}}$ symbolizes the space of proper, real-rational, $\hat{r}\times \hat{m}$ matrix-valued functions of $s\in \mathcal{C}$ which are analytic in $\mathcal{C}_+ \cup \mathcal{C}_-$. The state-space form of $\mathbf{P}_f(s)$ is given as
\begin{equation}
\begin{aligned}
\dot{\mathbf{x}}_f=&A_{f}\mathbf{x}_f+B_{f}\mathbf{\delta}_{f}\\
\mathbf{y}_f=&C\mathbf{x}_f
\end{aligned}
\label{eqp}
\end{equation}
where  $\mathbf{y}_f$ $\in$ $\mathbb{R}^{\hat{r}}$ and $C$ $\in$ $\mathbb{R}^{(\hat{r} \times \hat{n})}$
represent the output vector and output matrix, respectively. In the above state-space model, the number of outputs is less than the number of state variables ($\hat{r}$ $<$ $\hat{n}$) as  the autopilot hardware lacks sensors to measure every state variables.  Additionally, the number of unstable poles of $\mathbf{P}_f(s)$ and its  uncertain plant  is not the same. We now consider a finite set,  $\mathcal{P}=\{\mathbf{P}_1(s),\dots \mathbf{P}_f(s),\dots, \mathbf{P}_N(s)\} \subset	 \mathcal{RL_\infty}^{\hat{r} \times \hat{m}}$ contains these stabilizable LTI MIMO unstable adversely coupled  uncertain plants of a NAV. The unstable  plants of $\mathcal{P}$ also  have  a different number of unstable poles. 
The objective here is to find an   output feedback controller, $\mathbf{K} \in \mathbb{R}^{\hat{m} \times \hat{r}}$,   that achieves the following for all the plants belonging to $\mathcal{P}$.
\begin{enumerate}
\item Robust simultaneous stabilization.
\item Mode  decoupling.
\item Specified performance characteristics.
\end{enumerate}
\subsection{Simultaneous Stabilization with a Central Plant}\label{SS}
The method considered in this paper for simultaneously stabilizing a finite set of plants is to achieve a  controller that satisfies the sufficient condition for the existence of a simultaneous stabilizing controller associated with a central plant.  Here,  the sufficient condition given in \cite{jinjgcd} is employed. To explicate this sufficient condition, let the maximum $v$-gap metric of  $\mathbf{P}_i(s) \in \mathcal{P}$, $\epsilon_{\mathbf{P}_i}$ is given  by 
\begin{align}
\begin{split}
\epsilon_{\mathbf{P}_i}=\max\big\{&\delta_v\big(\mathbf{P}_{i}(j\omega),\mathbf{P}_f(j\omega)\big)~\big|~ \mathbf{P}_{i}(s), \mathbf{P}_f(s) \\& \quad \in \mathcal{P}~~ \forall~f \in \{1,2,\dots,N\}  \big\}  
\end{split}
\label{epi}
\end{align}
Following this, let the set, $\bar{\epsilon}$ contains the maximum $v$-gap metrics of all the plants in  $\mathcal{P}$. 
If  $\mathbf{P}_{cp}(s)$ is the central plant  of $\mathcal{P}$ with a maximum $v$-gap metric of $\epsilon_{\mathbf{P}_{cp}}$, then from the definition of central plant, the    $\epsilon_{\mathbf{P}_{cp}}$ = $\min~\bar{\epsilon}$. Thereupon, the sufficient condition  for the existence of a simultaneous stabilizing controller associated with $\mathbf{P}_{cp}(s)$ is given by \cite{jinjgcd}
\begin{align}
b_{\mathbf{P}_{cp},\mathbf{K}}>\epsilon_{\mathbf{P}_{cp}}=\min~\bar{\epsilon}
\label{sscondk1}
\end{align} 
To solve (\ref{sscondk1}), the  identification of $\mathbf{P}_{cp}(s)$ and $\epsilon_{\mathbf{P}_{cp}}$ is needed. This is accomplished by performing $v-$gap metric analysis  that has the following steps At first, find the values of maximum $v-$gap metrics of all the plants  using $\max\big\{\delta_v\big(\mathbf{P}_i(j\omega),\mathbf{P}_f(j\omega)\big)~\big|~ \mathbf{P}_i(s), \mathbf{P}_f(s) \in \mathcal{P}~ \forall~f \in \{1,2,\dots,N\}  \big\}$. Among these values, the smallest value gives $\epsilon_{\mathbf{P}_{cp}}$. Also, the plant  associated with this value is the central plant. The detailed explanation of the SS with a central plant  is given in the supporting material. \par
When the maximum $v-$gap metric of the central plant becomes smaller and closure to zero, similar closed-loop characteristics can be provided to all plants in $\mathcal{P}$ using a simultaneous stabilizing controller that satisfies (\ref{sscondk1}). The maximum $v-$gap metric of the central plant  is reduced and the frequency characteristics of the plants in $\mathcal{P}$ is  improved by cascading these plants with suitable pre and post compensators, $\mathbf{W_{in}}(s)$ and $\mathbf{W_{ot}}(s)$, respectively \cite{jinjgcd}. 
The problem of simultaneous reduction of the maximum $v-$gap metric of the central plant  and the improvement of frequency characteristics of all plants in $\mathcal{P}$ is termed as the simultaneous closeness-performance enhanced (SCP)  problem. Let $\kappa= \big\{ \mathbf{\tilde{P}}_f(s) ~ \big| ~ \mathbf{\tilde{P}}_f(s) \in \mathcal{RL_\infty}^{\hat{r} \times \hat{m}}, \mathbf{\tilde{P}}_f(s)= \mathbf{W_{ot}}(s)\mathbf{P}_f(s)\mathbf{W_{in}}(s),~$ $\mathbf{P}_f(s)$ $\in \mathcal{P}$, $\mathbf{W_{ot}}(s)$ $\in$ $\mathcal{RH_\infty}^{\hat{r} \times \hat{r}}$,   $\mathbf{W_{in}}(s)$ $\in$ $\mathcal{RH_\infty}^{\hat{m} \times \hat{m}}$, $\forall~f \in \{1,2,\dots,N\}\big\}$. Now, the SCP problem is stated as the following. Find  $\mathbf{W_{in}}(s)$ and $\mathbf{W_{ot}}(s)$ such that the following are achieved.
\begin{enumerate}
 \item The maximum $v-$gap metric of the central plant  of $\kappa$ (SCP central plant)  less than the maximum $v-$gap metric of the central plant of $\mathcal{P}$. 
 \item $\mathbf{W_{in}}(s)$ and $\mathbf{W_{ot}}(s)$  induce desired frequency characteristics on all the plants (performance enhanced plants) that belongs to $\kappa$.  
\end{enumerate}
Let $\hat{\epsilon}$ is defined as
\begin{align}
\begin{split}
\hat{\epsilon}&={}\big\{ \max\big\{\delta_v\big(\mathbf{W_{ot}}(j\omega)\mathbf{P}_1(j\omega)\mathbf{W_{in}}(j\omega),\mathbf{W_{ot}}(j\omega)\mathbf{P}_f(j\omega)\\ &\qquad\mathbf{W_{in}}(j\omega)\big)~|~ \forall~f \in \{1,2,..,N\}\big\},   \dots,   \max\big\{\delta_v\big(\mathbf{W_{ot}}(j\omega)\\ &\qquad \mathbf{P}_{N}(j\omega)\mathbf{W_{in}}(j\omega),\mathbf{W_{ot}}(j\omega)\mathbf{P}_f(j\omega)\mathbf{W_{in}}(j\omega)\big)~|~\forall~f \in \\ &\qquad \{1,2,..,N\}\big\} \big\}
\end{split}
\label{Lpi}
\end{align}
Then, the optimization problem to find the solution of the SCP  problem is given as  \cite{jinjgcd}
\begin{equation}
\begin{aligned}
& \underset{\textbf{Q}}{\text{minimize}}
& & J_1=\min~\hat{\epsilon}\\
& \text{subject to}
& & \text{1) \textit{Bound constraints on the coefficients of  pre} }\\
& & &\text{~~~~\textit{and post compensators}} \\ 
& & & \text{2) \textit{No pole-zero cancellation between} }\\
& & & \text{~~~\textit{compensators and the plants of $\mathcal{P}$}}
\end{aligned}
\label{pbm1}
\end{equation}
In    (\ref{pbm1}), $\mathbf{Q}$  represents the set that contains the coefficients of $\mathbf{W_{in}}(s)$ and $\mathbf{W_{ot}}(s)$.
 The bound constraints on the coefficients of  compensators are to provide desired frequency characteristics to the plants of $\mathcal{P}$. These constraints prevent the minimization of  $J_1$  with any $\mathbf{W_{in}}(s)$ and $\mathbf{W_{ot}}(s)$ that degrade the frequency characteristics of all the augmented plants. The performance index of (\ref{pbm1}) is non-convex and non-smooth. Hence, an iterative algorithm is required to obtain a solution. Note that the pre/post compensators are physically present in the closed-loop. Therefore, these compensators need to be appended with the hardware. 

\subsection{Appending Decoupling and  Robustness Features to a Simultaneously Stabilizing   Output Feedback Controller, $\mathbf{K}$}\label{RSSOFC}

\subsubsection{Appending Decoupling to  $\mathbf{K}$}
To explain the method that appends decoupling features  to $\mathbf{K}$, let  us consider  $\mathbf{W_{in}}(s)$ and $\mathbf{W_{ot}}(s)$ as the solution of  (\ref{pbm1}) and $\mathbf{\tilde{P}}_{cp}(s)=\mathbf{W_{ot}}(s)\mathbf{P}_i(s)\mathbf{W_{in}}(s)$ as the SCP central plant with  maximum $v-$gap metric of $\bar{J}_1$. 
Let the generalized stability margin of $\mathbf{\tilde{P}}_{cp}(s)$ is given as
\begin{equation}
b_{\mathbf{\tilde{P}}_{cp},\mathbf{K}}=\sfrac{1}{\bigg|\bigg|\left[ \begin{array}{c} {\mathbf{\tilde{P}}_{cp}(s)} \\ {\mathbf{I}}  \end{array} \right ]\left(\begin{array}{c} {\mathbf{I-K\tilde{P}}_{cp}(s)}  \end{array} \right )^{-1}\left[ \begin{array}{cc} {\mathbf{-I}} & {\mathbf{K}}  \end{array} \right ]\bigg|\bigg|_\infty}
\label{LPssc1}
\end{equation}
The   sufficient condition associated with $\mathbf{\tilde{P}}_{cp}(s)$ for the SS  is given as 
\begin{align}
\begin{split}
b_{\mathbf{\tilde{P}}_{cp},\mathbf{K}}>\bar{J}_1
\end{split}
\label{LPssc}
\end{align}
\noindent The condition given in  (\ref{LPssc})  is achieved by minimizing the infinity norm of $b_{\mathbf{\tilde{P}}_{cp},\mathbf{K}}$ using a finite number of stabilizing  output feedback controllers of $\mathbf{\tilde{P}}_{cp}(s)$. To solve a RSSD problem,  these controllers need to provide  desired mode decoupling  to the CL plants besides minimizing the infinity norm. Eigenstructure assignment is a well know linear controller design technique that assigns the desired CL eigenvalues and eigenvectors. Using this technique, the modes of a dynamic system can be decoupled from the state variable responses by assigning zero value to the cross-coupled eigenvector elements of the desired eigenvectors.
Following the eigenstructure assignment technique described in \cite{patton}, $\hat{r}$ desired CL eigenvalues and $\hat{m}$ entries of the corresponding  eigenvectors are assigned with the output feedback controller, $\mathbf{K} \in \mathbb{R}^{\hat{m} \times \hat{r}}$ given by
\begin{align}
\mathbf{K}=W(CR)^{-1}
\label{eq:ei1}
\end{align}
\noindent where $C$ denotes the output matrix of  $\mathbf{\tilde{P}}_{cp}(s)$.  $W$=$[W_1,\dots, W_l,\dots,W_{\hat{r}}]$
and  $R$=$[R_1,\dots,R_l,\dots,R_{\hat{r}}]$ are defined as
\begin{align}
\intertext{if $\lambda_l$ is real, then } 
\left[ \begin{array}{c} {R_l} \\ {W_l}  \end{array} \right] &\in
    \{S_l : \left[ \begin{array}{cc} {A-\lambda_lI} & {B}  \end{array} \right] \left[ \begin{array}{c} {R_l} \\ {W_l}  \end{array} \right]=0 \} \nonumber  
  \intertext{if $\lambda_l$ is complex, then} 
\begin{split}
\left[ \begin{array}{c} {R_l} \\ {W_l}  \end{array} \right]& \in \{S_l:\left[ \begin{array}{llll} {A-\lambda_l^{re}I} & {\lambda_l^{im}I} &{B}&{0} \\{-\lambda_l^{im}I} & {A-\lambda_l^{re}I} &{0}&{B}\end{array} \right]\left[ \begin{array}{l} {R_l^{re}} \\ {R_l^{im}} \\{W_l^{re}}\\{W_l^{im}}  \end{array} \right]\\ & \quad =0 \} 
\end{split}
 \label{v_s}
\end{align}
\noindent where  
$R_l$=$[R_l^{re}|R_l^{im}]^T$ and 
$W_l$=$[W_l^{re}|W_l^{im}]^T$ when $\lambda_l$ is complex. Furthermore, in (\ref{v_s}), $A$ and $B$ represent system and input matrices of $\mathbf{\tilde{P}}_{cp}(s)$, respectively. Also, $\lambda_l$, $\lambda_l^{re}$, and $\lambda_l^{im}$ denote $l$th desired eigenvalue, real part of the $l$th desired eigenvalue, and imaginary part of the $l$th desired eigenvalue, respectively. $R_l$ and 
$W_l$ are computed using specific values assigned to $\hat{m}$ entries of   eigenvectors associated with a desired CL eigenvalue. For decoupling, these  $\hat{m}$ entries  are assigned with zero value. \par
The controller $\mathbf{K}$ given by  (\ref{eq:ei1}) depends on the allowable eigenvector space \cite{patton}. This  eigenvector space depends on the desired CL eigenvalues, system, input, and output matrices of the plant. For the given desired CL eigenvalues and eigenvectors, a single $\mathbf{K}$ is obtained. However,  a single $\mathbf{K}$ does not minimize the infinity norm given in   (\ref{LPssc1}). This issue is solved using the following approach. Usually, the desired CL eigenvalues are selected based on  specific characteristics. For example, one such characteristic is the damping ratio. If the desired CL eigenvalues require a damping ratio greater than or equal to $\zeta_{de}$, then these eigenvalues lie in a region defined by the constant damping ratio line. Following this, a region, $S_1$ in $\mathbb{C}_-$ is predefined using the specific characteristics required for the desired CL eigenvalues.  There will be  an infinite number of desired CL eigenvalues in this region.
Now, obtain the $\mathbf{K}$ using a search and minimization algorithm that performs the following functions.
\begin{enumerate}
\item Search in  $S_1$ and obtain  the desired CL eigenvalues, \textrm{$\lambda_l$ $\in$ $S_1$ $\forall$ $l$ $\in$ \{1,2,\dots,$\hat{r}$\}}.
\item Compute  $\mathbf{K}=W(CR)^{-1}$ using the desired CL eigenvalues  from $S_1$.
\item Check whether $\mathbf{K}=W(CR)^{-1}$  stabilizes  $\mathbf{\tilde{P}}_{cp}(s)$. For this purpose, examine whether the remaining $n-\hat{r}$ ($n$ is the number of states of $\mathbf{\tilde{P}}_{cp}(s)$) CL eigenvalues, \textrm{$\lambda_i$  $\forall$ $i$ $\in$ \{$\hat{r}$+1, $\hat{r}$+2,\dots,$n$\}}, are in $S_1$. If \textrm{$\lambda_i$ $\in$ $S_1$ $\forall$ $i$ $\in$ \{$\hat{r}$+1, $\hat{r}$+2,\dots,$n$\}}, then $\mathbf{K}=W(CR)^{-1}$  stabilizes  $\mathbf{\tilde{P}}_{cp}(s)$.
  Therefore, if \textrm{$\lambda_i$ $\in$ $S_1$ $\forall$ $i$ $\in$ \{$\hat{r}$+1, $\hat{r}$+2,\dots,$n$\}}, then go to 4. Otherwise, go to 1.
\item Minimize infinity norm  given  in  (\ref{LPssc1}) using $\mathbf{K}$.
\item Repeat 1, 2, 3, and 4 until $\mathbf{K}$ satisfies the
sufficient condition for SS given in (\ref{LPssc}).
\end{enumerate}
 In brief, if the search and minimizing algorithm obtains a $\mathbf{K}=W(CR)^{-1}$  that achieves (\ref{LPssc}) with all $\lambda(A+B\mathbf{K}C) \in S_1$, then $\mathbf{K}$ provides desired decoupling to the CL plant  of $\mathbf{\tilde{P}}_{cp}(s)$ and SS to all the plants in $\kappa$.
Note that even with this $\mathbf{K}$, when the CL characteristics of  $\mathbf{\tilde{P}}_{cp}(s)$ and the CL characteristics of remaining plants of $\kappa$ are dissimilar, then a new solution for the SCP problem is required. This solution needs to provide a new  $\mathbf{\tilde{P}}_{cp}(s)$ with a maximum $v-$gap metric that should be less than the previous $\bar{J}_1$. This is because the CL characteristics of  $\mathbf{\tilde{P}}_{cp}(s)$ and the remaining plants of $\kappa$ will be similar when $\bar{J}_1$ of $\mathbf{\tilde{P}}_{cp}(s)$ gets closer to zero.
\newtheorem{rem}{Remark}[section]
\begin{rem}
The term inside the infinity norm of (\ref{LPssc1}) is the CL transfer function of $\mathbf{\tilde{P}}_{cp}(s)$ from  exogenous inputs to controlled variables.
Hence, minimizing the infinity norm
 also improves tracking performance, disturbance rejection, and noise rejection of the CL plant.   Following this, if the maximum $v-$gap metric of  $\mathbf{\tilde{P}}_{cp}(s)$ is closer to zero, then the closed-loop performance characteristics of  all the plants in $\kappa$ will be similar to the closed-loop performance characteristics of $\mathbf{\tilde{P}}_{cp}(s)$   achieved by minimizing the infinity norm of (\ref{LPssc1}).
\end{rem}
 \subsubsection{Appending Robustness to a Simultaneously Stabilizing Decoupling  Output Feedback Controller, $\mathbf{K}$}
An  output feedback controller, $\mathbf{K}=W(CR)^{-1}$  that satisfies (\ref{LPssc})  accomplishes SS and decoupling. Besides this, an  RSSD  output feedback controller needs to provide robustness to each plants belonging to $\kappa$.  The concept of robust SS is illustrated in Fig. \ref{fig:RSS}. This figure suggests that a single controller must simultaneously stabilize all plants belonging to $\kappa$, but also need to stabilize each plant  against certain normalized right coprime factor perturbations. We now propose the following lemma that illustrates the robustness features of a simultaneous stabilizing decoupling controller.
\begin{figure}[!h]
\begin{center}
\includegraphics[width=3.2in, height=1.8in]{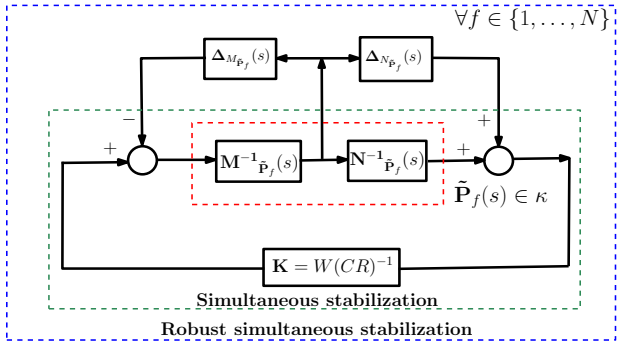}
\caption{\textbf{Concept of robust simultaneous stabilization}}
\label{fig:RSS}
\end{center}
\end{figure}
\newtheorem{lem}{Lemma}[section]
\begin{lem}
Let $W$ and $R$ are chosen for achieving the desired decoupling.  Then, the   output feedback controller, $\mathbf{K}=W(CR)^{-1}$, of $\mathbf{\tilde{P}}_{cp}(s)$ accomplishes robust simultaneous stabilization of N plants  belonging to $\xi=\big\{  \mathbf{\tilde{P}}(s)~ \big|~ \mathbf{\tilde{P}}(s) \in \kappa, \delta_v(\mathbf{\tilde{P}}_{cp}(j\omega),\mathbf{\tilde{P}}(j\omega)) \leq \bar{J}_1 ~\mathbf{\tilde{P}}_{cp}(s) \in \kappa \big\}$ and provides exact desired decoupling to the CL plant  of $\mathbf{\tilde{P}}_{cp}(s)$, if the following conditions are satisfied.
\begin{enumerate}[label={[\arabic*]}]
\item \textrm{$\lambda_l$ $\in$ $S_1$ $\forall$ $l$ $\in$ \{1,2,\dots,$\hat{r}$\}}
\item \textrm{$\lambda_i$ $\in$ $S_1$ $\forall$ $i$ $\in$ \{$\hat{r}$+1,$\hat{r}$+2,\dots,n\}}
\item $b_{\mathbf{\tilde{P}}_{cp},\mathbf{K}}>\bar{J}_1$
\end{enumerate}
\label{lemma2}
\end{lem}
\begin{proof}
$\mathbf{K}$ stabilizes $\mathbf{\tilde{P}}_{cp}(s)$ when the conditions [1] and [2] are met.  Concurrently, $\mathbf{K}$ provides exact desired decoupling to $\mathbf{\tilde{P}}_{cp}(s)$ and  simultaneously stabilizes N plants   belonging to $\xi$ when $\mathbf{K}$ achieves all  three conditions. Now to prove  $\mathbf{K}$  achieves robust simultaneous stabilization of N plants,  let us define the  ball of plants, $B(\mathbf{\tilde{P}}_{cp}(s),\bar{J}_1)$,  as
\begin{align}
\begin{split}
B(\mathbf{\tilde{P}}_{cp}(s),\bar{J}_1)={}&\big\{\textrm{${\mathbf{\breve{P}}}(s)$ $\big|$ $\delta_v$\big($\mathbf{\tilde{P}}_{cp}(j\omega)$, ${\mathbf{\breve{P}}}(j\omega)$\big) $\leq$ $\bar{J}_1$}, \\ & \qquad \textrm{$\mathbf{\tilde{P}}_{cp}(s)$ $\in$ $\xi$,  ${\mathbf{\breve{P}}}(s)$ $\in$ $\mathcal{RL_\infty}^{\hat{r} \times \hat{m}}$}\big\}
\end{split}
\label{ball2}
\end{align}
When $\mathbf{K}$ satisfies all the three conditions, any plant   belonging to  $\partial B$ (boundary of $B(\mathbf{\tilde{P}}_{cp}(s), \bar{J}_1)$) is also stabilized by $\mathbf{K}$ as  $v-$gap metrics between the plants  belonging to $\partial B$ and $\mathbf{\tilde{P}}_{cp}(s)$ are $\bar{J}_1$.   $v-$gap metrics between $\mathbf{\tilde{P}}_{cp}(s)$ and the  plants in $Int($B$)$ (interior of $B(\mathbf{\tilde{P}}_{cp}(s), \bar{J}_1)$)   are less than $\bar{J}_1$ as    $\bar{J}_1$  is the  maximum $v-$gap metric of $\mathbf{\tilde{P}}_{cp}(s)$. Therefore,  $\mathbf{K}$ simultaneously stabilizes all the plants belong to $B(\mathbf{\tilde{P}}_{cp}(s),\bar{J}_1)$ if  the conditions [1], [2], and [3] are met.  As $\xi$ $\subset$ $B(\mathbf{\tilde{P}}_{cp}(s), \bar{J}_1)$, $\mathbf{K}$ simultaneously stabilizes all the N plants. Even if any plant, $\mathbf{\tilde{P}}(s)$ $\in$ $\xi$ is perturbed to form the plant, $\underline{\mathbf{P}}(s)$  $\in$ $B(\mathbf{\tilde{P}}_{cp}(s), \bar{J}_1)\backslash \xi$ is also stabilized by $\mathbf{K}$. Therefore, $\mathbf{K}$ achieves robust simultaneous stabilization of  N  plants belonging to $\xi$. This establishes the proof.
\end{proof}
The three conditions of Lemma \ref{lemma2} imply sufficient conditions for the existence of a RSSD  output feedback controller. Moreover,  the structure of $\xi$  suggests  any $\mathbf{K}$ that satisfies the three conditions of Lemma~\ref{lemma2} also achieves robust simultaneous stabilization of N plants of $\mathcal{P}$.
\subsection{NN-RSSD   Algorithm}
 An iterative algorithm is required for solving a RSSD problem due to the NP-hard nature of  output feedback and SS problems.  The controller generated by the iterative algorithm becomes the solution of a RSSD problem once this controller satisfies three conditions of Lemma \ref{lemma2}.
Using these conditions, we now recast a RSSD problem as an optimization problem that is  given as
\begin{equation}
\begin{aligned}
& \underset{\mathbf{K}}{\text{minimize}}
& &J_2=\bigg|\bigg|\left[ \begin{array}{c} {\mathbf{\tilde{P}}_{cp}(s)} \\ {\mathbf{I}}  \end{array} \right ]\left(\begin{array}{c} {\mathbf{I-K\tilde{P}}_{cp}(s)}  \end{array} \right )^{-1}\left[ \begin{array}{cc} {\mathbf{-I}} & {\mathbf{K}}  \end{array} \right ]\bigg|\bigg|_\infty\\
& \text{s.t.} 
& &  \textrm{$\lambda_l$ $\in$ $S_1$ $\forall$ $l$ $\in$ \{1,2,\dots,$\hat{r}$\}}\\
&&& \mathbf{K}=W(CR)^{-1}\\
&&& \textrm{$\lambda_i$ $\in$ $S_1$ $\forall$ $i$ $\in$ \{$\hat{r}$+1,$\hat{r}$+2,\dots,n\}}
\end{aligned}
\label{pbm2}
\end{equation}
The solution of a RSSD problem is a    $\mathbf{K}  \in \mathbb{R}^{\hat{m} \times \hat{r}}$ that reduces $J_2$  below $\bar{J}_1$ with   all the constraints of (\ref{pbm2}) satisfied.   A  genetic algorithm based offline iterative algorithm referred to as the NN-RSSD  algorithm is developed for sequentially solving the optimization problems associated with the SCP  problem and the RSSD problem. The genetic algorithm is employed in  the iterative algorithm as the optimization problem associated with the SCP  problem is non-convex and non-smooth.  
At first, the NN-RSSD  algorithm solves the optimization problem of the SCP  problem given in (\ref{pbm1}). 
Then,  the optimization problem of the RSSD problem given in (\ref{pbm2}) is solved utilizing the solution of the SCP  problem.
For this purpose, the NN-RSSD algorithm has two population-based genetic algorithm (GA) solvers, namely, GA-SCP and GA-RSSD. GA-SCP solves the optimization problem of the SCP  problem, whereas  GA-RSSD solves the optimization problem of the RSSD problem. In GA-SCP and GA-RSSD, 
 GA employs the following steps.
 \begin{enumerate}
 \item Randomly generate initial  values of search variables.
 \item Compute the value of fitness function using  the feasible values among the values of search variables.
 \item Generate new  values of search variables by examining  the fitness value of search variables and   applying genetic operators on the values of search variables.
 \item Repeat 2 and 3 until the stopping criterion is satisfied.
 \end{enumerate}
\textit{Search Variables:} The search variables of  GA-SCP   are the coefficients of $\mathbf{W_{ot}}(s)$ and $\mathbf{W_{in}}(s)$.
The feasible values of these search variables are those which  satisfy all the constraints of the  problem given in (\ref{pbm1}).
In the  NN-RSSD  algorithm, the chosen desired eigenvector element for mode decoupling is not assigned a zero value, but it is confined in a bounded set and used as a search variable by GA-RSSD. This reduces the failure of the eigenstructure assignment algorithm in computing the controller gain using (\ref{eq:ei1}) when all the chosen  desired eigenvector elements are assigned with zero value. Further, the availability of the   chosen  desired eigenvector element as  search variables  contributes additional freedom in the computation of the controller defined by (\ref{eq:ei1}).  The bound on a chosen  desired eigenvector element  is determined by the magnitude of unfavorable dominance of a mode on the state variable. Also,   the search variables of GA-RSSD include $\hat{r}$ desired CL eigenvalues.
The feasible values of $\hat{r}$ eigenvalues and chosen desired eigenvectors are those belonging to $S_1$ and the user-defined bounded set, respectively.

\textit{Fitness functions:} The fitness function of GA-SCP is the performance index, $J_1$, of  the optimization  problem given in (\ref{pbm1}). 
 The fitness function of GA-RSSD is the performance index, $J_2$, of  the optimization problem given in (\ref{pbm2}).
  
\textit{Termination Conditions:} The NN-RSSD algorithm  terminates when the number of generation of GA-SCP exceeds its maximum value or when GA-RSSD generates a feasible RSSD  output feedback controller.\par
The working of the NN-RSSD algorithm is explained as follows. At first,  GA randomly obtains the initial values of the search variables of GA-SCP. Then, the fitness function of GA-SCP is evaluated, as shown in Algorithm-1. During this fitness evaluation, if the  obtained solution satisfies an adaptive constraint, then the GA-SCP invokes GA-RSSD   and passes $\mathbf{\tilde{P}}_{cp}(s)$ and its maximum $v-$gap metric to the GA-RSSD. The adaptive constraint employed here is \textit{fitness value of GA-SCP} $<$ $\bar{J}_1$.  During the fitness function evaluation of the GA-RSSD, if the value of the fitness function satisfies the condition, $J_2 < \frac{1}{\bar{J}_1}$,  then that controller becomes the feasible controller. Thereafter, the NN-RSSD algorithm terminates. When the GA-RSSD could not find a feasible controller for a $\mathbf{\tilde{P}}_{cp}(s)$, then $\bar{J}_1$ is assigned with the fitness value that previously satisfied the adaptive constraint.  Note that, new $\bar{J}_1$ $<$ old $\bar{J}_1$ enables GA-SCP to search for newer $\mathbf{\tilde{P}}_{cp}(s)$ that has a maximum $v-$gap metric which is  less than new $\bar{J}_1$. This enhances the possibility of obtaining a feasible RSSD  output feedback controller. During the fitness evaluation of GA-SCP, if the solution obtained does not satisfy an adaptive constraint, then GA searches for new feasible values. Note that for computing $\mathbf{K}$ requires the inverse of  $CR$. If the condition number of $CR$, $\kappa(CR)$, is infinity, then $CR$ singular. Following this, the existence of the inverse of  $CR$ is guaranteed by satisfying the condition, $\kappa(CR)< 10^5$.  If this condition is satisfied,  then compute $\mathbf{K}=W(CR)^{-1}$ else GA obtain new feasible values of GA-RSSD's search variables. The pseudocode of the NN-RSSD algorithm is given in Algorithm-1 and Algorithm-2.  
 
\section{Design and Performance  Evaluation  of  the RSSD  Output Feedback Controller for a Fixed-Wing NAV}\label{DEex}
This section discusses the design of a RSSD output feedback controller using the NN-RSSD  algorithm for eight unstable MIMO adversely coupled  plants of the NAV mentioned in \cite{jin}. More details about the nonlinear and linear plants of the NAV,    desired eigenvector elements,  and  bounds  of pre/post compensators are given in  the supporting material.
 Among the output variables of these plants (pitch angle ($\theta$), pitch rate ($q$),  yaw rate ($r$), roll rate ($p$), and roll angle ($\phi$)),  the controlled  variables are $\theta$ and $\phi$.  
The frequency characteristics of the eight plants require modifications as they do not have the desired characteristics.   To incorporate necessary frequency characteristics on the output sensitivity function, we require $\mathbf{W_{ot}}(s)$ and $\mathbf{W_{in}}(s)$ that  satisfy $\underline{\sigma}(\mathbf{W_{ot}(0)}\mathbf{P}_f(0)\mathbf{W_{in}(0)}) >~6$~dB  and $\underline{\sigma}(\mathbf{W_{ot}}(j\omega)\mathbf{P}_f(j\omega)\mathbf{W_{in}}(j\omega))~>~0$~dB  up to the desired frequency for all $f \in \{1,2,\dots,8\}$. Here, the desired frequency ranges from 15.28~rad/s   to 53.3~rad/s. $\mathbf{P}_{2}(s)$ among the eight plants is identified as the central plant as it has the smallest maximum $v-$gap metric of 0.4489 ($\min~\bar{\epsilon}$). 
\begin{algorithm}[!h]
\caption{Pseudocode of NN-Algorithm}
\begin{algorithmic}[1]
   \State Input: $\mathcal{P}$, $\bar{J}_1$,  $n$, and $\hat{r}$
    \If{number of generation of GA-SCP$\leq$ maximum value}
    \State GA obtain feasible values of GA-SCP's search variables
   \State Compute: $\mathbf{W_{in}}(s)$ and $\mathbf{W_{ot}}(s)$ using (13)-(16) given in \cite{jinjgcd}.
    \State Compute: $J_1$ for fitness evaluation 
    \If{Fitness value $<$ $\bar{J}_1$}
        \State Obtain: $\mathbf{\tilde{P}}_{cp}(s)$ by performing $v-$gap metric analysis on $\xi$
        \State Set $\bar{J}_1$=Fitness value
        \State \Call{GA-RSSD}{$\mathbf{\tilde{P}}_{cp}(s)$, $\bar{J}_1$, $n$, $\hat{r}$}
         \State go to 2
        \Else
        \State go to 2
     \EndIf
      \Else
        \State Exit 
     \EndIf
 \label{Alg-1}
\end{algorithmic}
\end{algorithm}
\begin{algorithm}[!h]
	\caption{Pseudocode of GA-RSSD}
	\begin{algorithmic}[1]
		\Function{GA-RSSD}{$\mathbf{\tilde{P}}_{cp}(s)$, $\bar{J}_1$, $n$, $\hat{r}$}
		\If{number of generation of GA-RSSD$\leq$ maximum value}
		\State GA obtain feasible values of GA-RSSD's search variables
		\State Compute: $W$ and $R$ (using $A$ and $B$ of $\mathbf{\tilde{P}}_{cp}(s)$ and feasible values of search variables of GA-RSSD)
		\State Compute: $CR$ and $\kappa(CR)$
		\If{ $\kappa(CR)$ $<$ 10$^5$}
		\State Compute: $\mathbf{K}=W(CR)^{-1}$
		\If{$\lambda_i(A+B\mathbf{K}C)~ \forall~ i \in \{1,2,\dots,n\} \in  S_1$} 
		\State Compute: $J_2$ for fitness evaluation
		\If{Fitness value $<$ $\frac{1}{\bar{J}_1}$}
		\State Controller found ($\mathbf{K}$) and Exit
		\Else
		\State go to 2
		\EndIf
		\Else
		\State go to 2
		\EndIf 
		\Else
		\State go to 2
		\EndIf
		\Else
		\State \textbf{return}
		\EndIf
		\EndFunction
		\label{Alg-11}
	\end{algorithmic}
\end{algorithm} 
\subsection{Flight Control System Design Specifications and Controller Gain Synthesis}
\begin{enumerate}
\item \textit{Damping Ratio Specifications:}
The minimum required damping ratios for all the oscillatory modes of the NAV are taken as  0.30.
\item \textit{Tracking Specifications:} The tracking error in the pitch angle and the roll angle of the NAV in the absence of wind disturbances concerning the reference signal are ${\pm0.0087}$~rad  and ${\pm0.0175}$~rad, respectively. Also, in the presence of wind disturbances, root mean square deviation (RMS) of the pitch angle and roll angle need not exceed 0.0873~rad  and 0.175~rad, respectively.
\item \textit{Robustness Requirements:}
An uncertainty level of 40~$\%$ is assigned to all the static, dynamic, and control derivatives of the NAV. Concurrently, for each CL plant, the controller has to provide a  multiloop disk-based gain  margin (MDGM)   and  phase  margin (MDPM) of  $\pm$3~dB and $\pm$20~deg, respectively.
\item \textit{Desired eigenvalues and eigenvectors:}
The eigenvalues of the coupled spiral, coupled Dutch roll, and coupled short period modes of the NAV are assigned as these eigenvalues do not possess the desired stability and performance characteristics. For more details about these modes and their coupling characteristics, one should refer \cite{jin}. The bounds on the chosen  desired eigenvector elements for mode decoupling  are given in Table \ref{tab:dec1}.
In Table \ref{tab:dec1}, $\bar{u}$, $\bar{w}$, $\bar{r}$, $\bar{p}$, and $\bar{\phi}$ denote the eigenvector elements associated with the state variables, $u$ (translational velocity in $\mathbf{x}$ body-axis), $w$ (translational velocity in $\mathbf{w}$ body-axis), $r$, $p$, and $\phi$, respectively.  
\end{enumerate}
 \begin{table}[h!]
\centering
\caption{\label{tab:dec1} Bounds on the chosen  desired eigenvector elements}
\begin{tabular}{llll}
\hline
Mode &Eigenvector&Bound\\
 &Element&\\
\hline
Coupled Spiral&$\bar{u}, \bar{w}, \bar{r}  $&$\pm$0.01, $\pm$0.01,$\pm$0.01\\
Coupled Dutch &$\bar{p}$&[$\pm$0.049]$\pm$[$\pm$0.049]$j$\\
roll&$\bar{\phi}$&[$\pm$0.1]$\pm$[$\pm$0.08]$j$\\
Coupled Short&$\bar{r}$&[$\pm$0.015]$\pm$[$\pm$0.15]$j$\\
period&&&\\
\hline
\end{tabular}
\end{table} 

 The NN-RSSD  algorithm  is executed with inputs, $\hat{r}=5$, $\hat{m}=3$, $n=19$, $\bar{J}_1$=0.4489, $\mathcal{P}$,  the bounds on the coefficients of $\mathbf{W_{ot}}(s)$ and $\mathbf{W_{in}}(s)$,  the bounds on the chosen desired eigenvector elements,   the maximum number of generations of GA-SCP (20) and GA-RSSD (1000), and  $S_1$ defined   for the damping ratio of 0.3. NN-RSSD  algorithm terminates at the 388th generation of GA-RSSD. Subsequently,   the feasible compensators  are given as

\begin{small}
\begin{align}
\mathbf{W_{in}}(s)=diag\bigg[\frac{1.0s+7.36}{0.007s+10.1},\frac{14.71s+59.78}{0.002s+0.099},\frac{0.992s+5.091}{0.0053s+11.89}\bigg]\label{wine}
\end{align}
\end{small}
\begin{small}
\begin{align}
\begin{split}\label{wone}
\mathbf{W_{ot}}(s)={}&diag\bigg[\frac{0.61s+11}{1.6s+1.6},\frac{0.23s+29.58}{0.55s+0.36},\frac{0.813s+12.1}{7.996s+1.46},\\& \qquad \frac{0.91s+9.78}{0.331s+0.342}, \frac{0.78s+16.81}{1.41s+1.0}\bigg]
\end{split}
\end{align}
\end{small}
\noindent The feasible RSSD  output feedback controller gain is given as
\begin{equation}
\small{\mathbf{K}=\left[\begin{array}{ccccc}{1.13}&{0.78}&{0.26}&{0.14}&{0.13}\\{-1.07}&{  -0.78}&{ 0.74}&{0.009}&{0.54}\\{0.19}&{-0.06}&{-0.002}&{ 1.29}&{0.04}
\end{array}\right]}
\label{gain}
\end{equation}
Besides this, the SCP central plant is $\mathbf{\tilde{P}}_{cp}(s)$=$\mathbf{W_{ot}}(s)\mathbf{P}_3(s)\mathbf{W_{in}}(s)$ with the maximum $v-$gap metric of 0.3184. As this  value is less than 0.4489, one of the objectives of the pre and post compensators is achieved.  The frequency characteristics of  $\mathbf{\tilde{P}}_f(s)$ = $\mathbf{W_{ot}}(s)\mathbf{P}_f(s)\mathbf{W_{in}}(s)$ for all $f \in \{1,2,\dots,8\}$   are shown in Fig. \ref{fresigal1}. 
\begin{figure}[h!]
	\centering
	{\includegraphics[width=2.75in, height=1.75in]{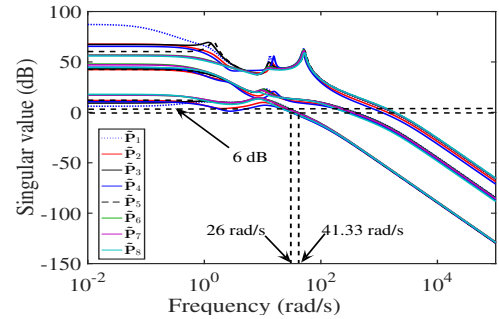}}
	\caption{Singular value plots  of the   plants}. 
	\label{fresigal1}
\end{figure} 
This figure indicates $\underline{\sigma}(\mathbf{\tilde{P}}_f(0))>6$~dB for all $f \in \{1,2,\dots,8\}$.   As illustrated in Fig. \ref{fresigal1}, the 0~dB crossing frequencies of all the performance enhanced  plants are within the desired limit as the range of 0~dB crossing frequencies of $\underline{\sigma}(\mathbf{\tilde{P}}_f(s))$ for all $f \in \{1,2,\dots,8\}$ are within 26-41.33~rad/s.  Hence,  $\mathbf{W_{in}}(s)$ and $\mathbf{W_{ot}}(s)$ given in  (\ref{wine}) and (\ref{wone}) satisfy all the design requirements with respect to the  performance enhanced  plants.
\subsection{Stability analysis}
The controller gain given in  (\ref{gain}) simultaneously stabilizes eight plants of the NAV as all  the eigenvalues of the CL plants, $\mathbf{\hat{H}}_f(s)=\frac{\mathbf{\tilde{P}}_f(s)\mathbf{K}}{\mathbf{I}-\mathbf{\tilde{P}}_f(s)\mathbf{K}}$ for all  $f \in \{1,2,\dots,8\}$ lie in $\mathcal{C}_-$ as shown in Fig. \ref{fig:p1}. 
\begin{figure}[h!]
	\begin{center}
		\includegraphics[width=2.75in, height=1.75in]{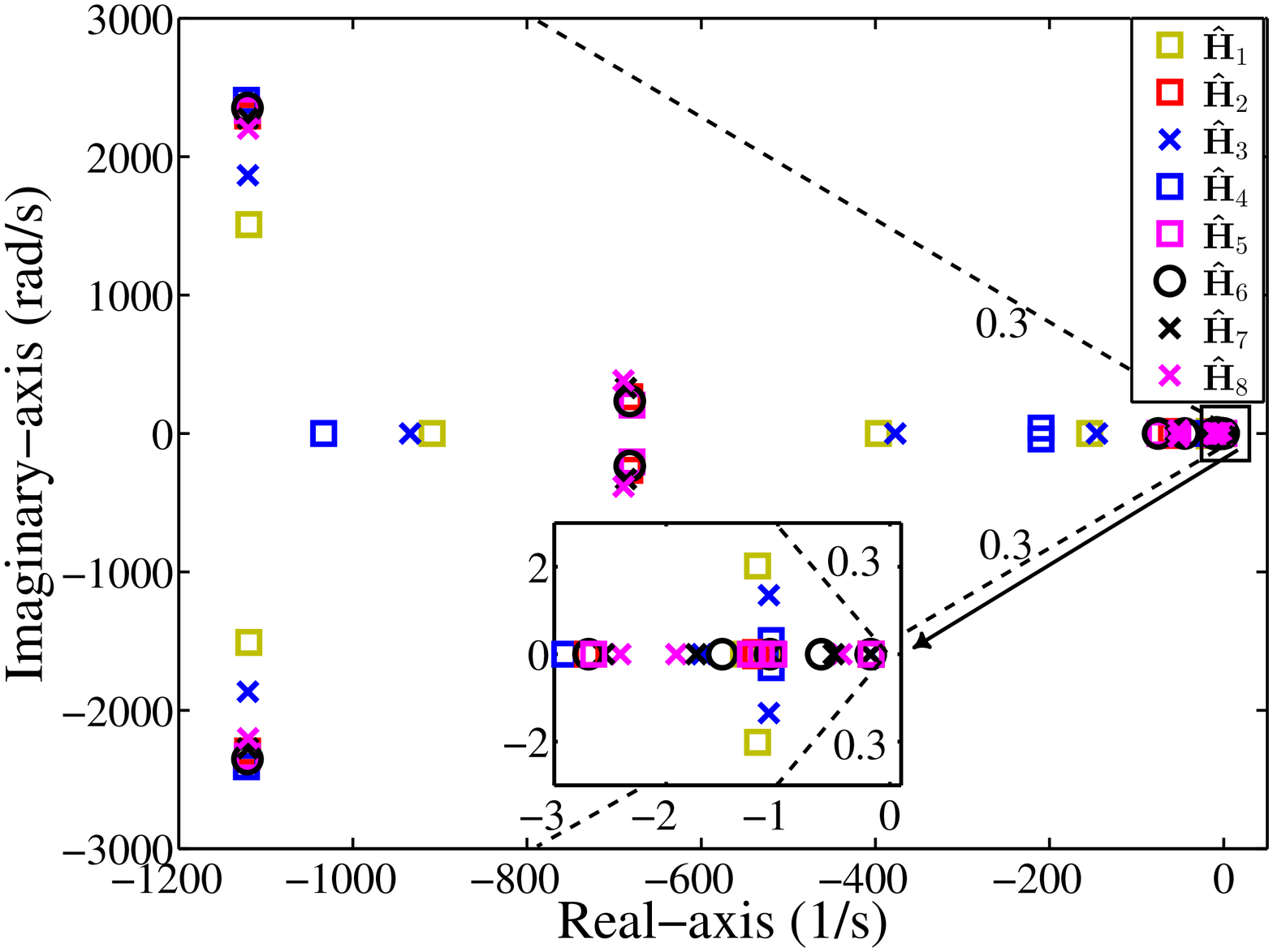}
		\caption{Eigenvalues of the   closed-loop plants of  NAV}
		\label{fig:p1}
	\end{center}
\end{figure}
This figure also indicates the damping ratios of all the CL eigenvalues  are greater than or equal to 0.3. 
Besides this, Table \ref{tab:disk1} suggest that all the eight CL plants have desired  MDGM and MDPM.
 \begin{table}[h!]
\centering
\caption{\label{tab:disk1} Multiloop disk-based gain (dB) and phase margins (deg) of the NAV}
\begin{tabular}{llllll}
\hline
Plant &\multicolumn{1}{c}{MDGM} &\multicolumn{1}{c}{MDPM}&Plant &\multicolumn{1}{c}{MDGM} &\multicolumn{1}{c}{MDPM}\\
\hline
$\mathbf{\hat{H}_1}(s)$&$\pm$4.30&$\pm$27.30&$\mathbf{\hat{H}_5}(s)$&$\pm$3.39&$\pm$21.79\\
$\mathbf{\hat{H}_2}(s)$&$\pm$3.42&$\pm$22.00&$\mathbf{\hat{H}_6}(s)$&$\pm$3.37&$\pm$21.69\\
$\mathbf{\hat{H}_3}(s)$&$\pm$3.88&$\pm$24.87&$\mathbf{\hat{H}_7}(s)$&$\pm$3.44&$\pm$22.10\\
$\mathbf{\hat{H}_4}(s)$&$\pm$3.33&$\pm$21.44&$\mathbf{\hat{H}_8}(s)$&$\pm$3.50&$\pm$22.51\\
\hline
\end{tabular}
\end{table}
The parametric and unmodeled dynamic uncertainties are represented by the  inverse input multiplicative and output multiplicative uncertainties, respectively. The tolerance level of the CL plants against the output multiplicative uncertainty is analyzed using the  plots of $\frac{1}{\overline{\sigma}(\mathbf{\tilde{P}}_{f}(j\omega)\mathbf{K}(\mathbf{I-\tilde{P}}_{f}(j\omega)\mathbf{K})^{-1})}$ for all $f \in \{1,2,\dots,8\}$ depicted in Fig. \ref{mulunall3}.
This figure depicts that the output multiplicative uncertainty tolerance levels increase with the frequency. This characteristic is required to tackle the output multiplicative uncertainty that occurs in the  high-frequency region. Further,  $\mathbf{\hat{H}}_3(s)$  has the worst-case tolerance level of about 51.8~$\%$ at 2.22~rad/s. This tolerance level is  in the low-frequency region (with respect to the bandwidth of the augmented plants), where the unmodeled dynamic   uncertainty does not occur. The tolerance level of the CL plant  against the inverse input multiplicative uncertainty is analyzed by studying the  plots of $\frac{1}{\overline{\sigma}(\mathbf{(I-K}\mathbf{\tilde{P}}_f(s))^{-1})}$ for all  $f \in \{1,2,\dots,8\}$ shown in Fig. \ref{invinmulunall3}. This figure shows $\mathbf{\hat{H}}_4(s)$ has the worst-case inverse input multiplicative uncertainty tolerance level   of 62~$\%$ at  2987~rad/s. This worst-case tolerance level occurs in the high-frequency region, where parametric uncertainty does not exist. The tolerance levels of all the CL plants are greater than 79~$\%$ when the frequency is below 41.33~rad/s (maximum bandwidth of augmented plant). This satisfies the  tolerance level requirement  against the parametric uncertainties.
\begin{figure}[h!]
\begin{center}
\includegraphics[width=2.75in, height=1.75in]{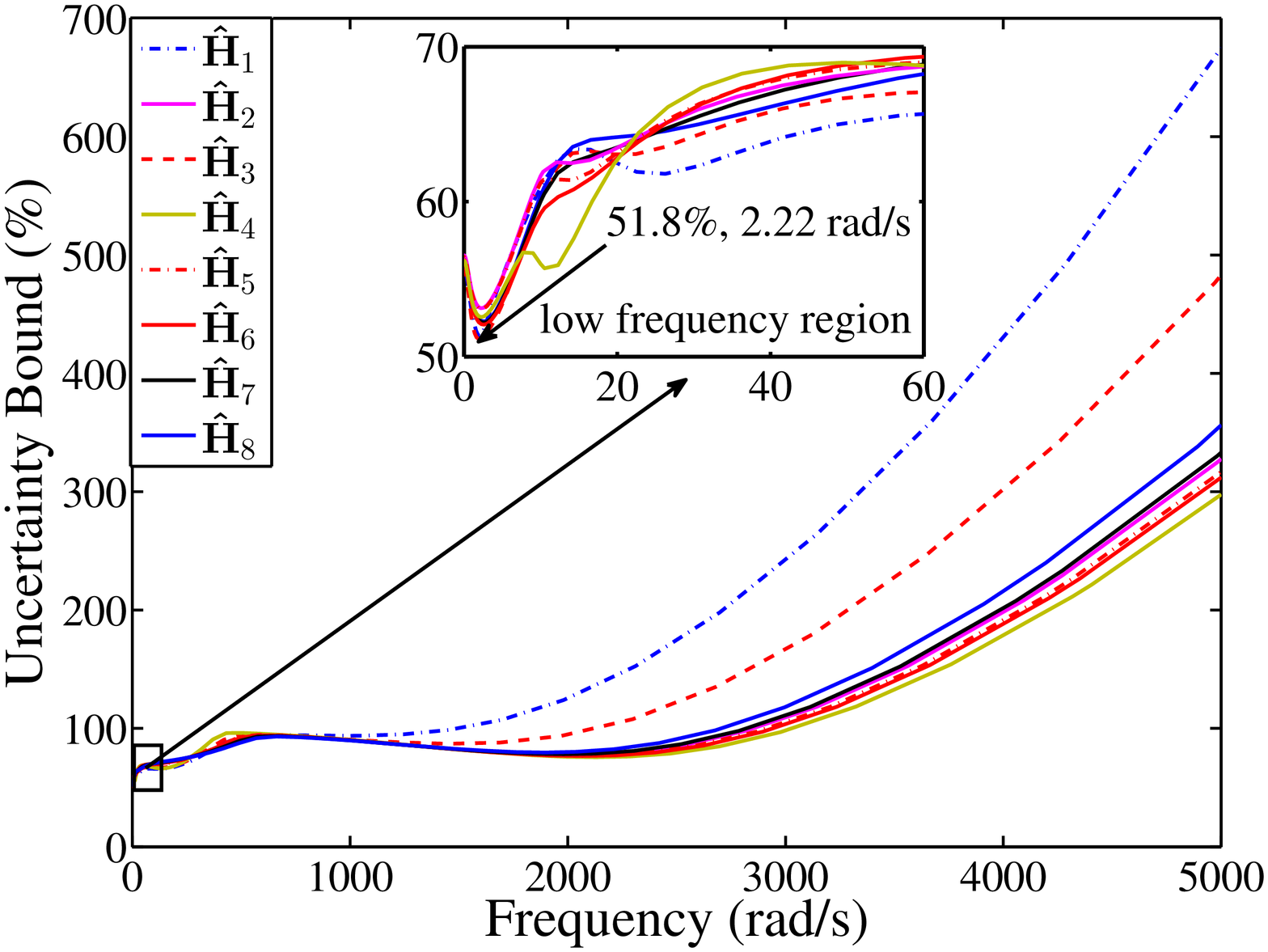}
\caption{Output multiplicative uncertainty bound}
\label{mulunall3}
\end{center}
\end{figure}
\begin{figure}[h!]
\begin{center}
\includegraphics[width=2.75in, height=1.75in]{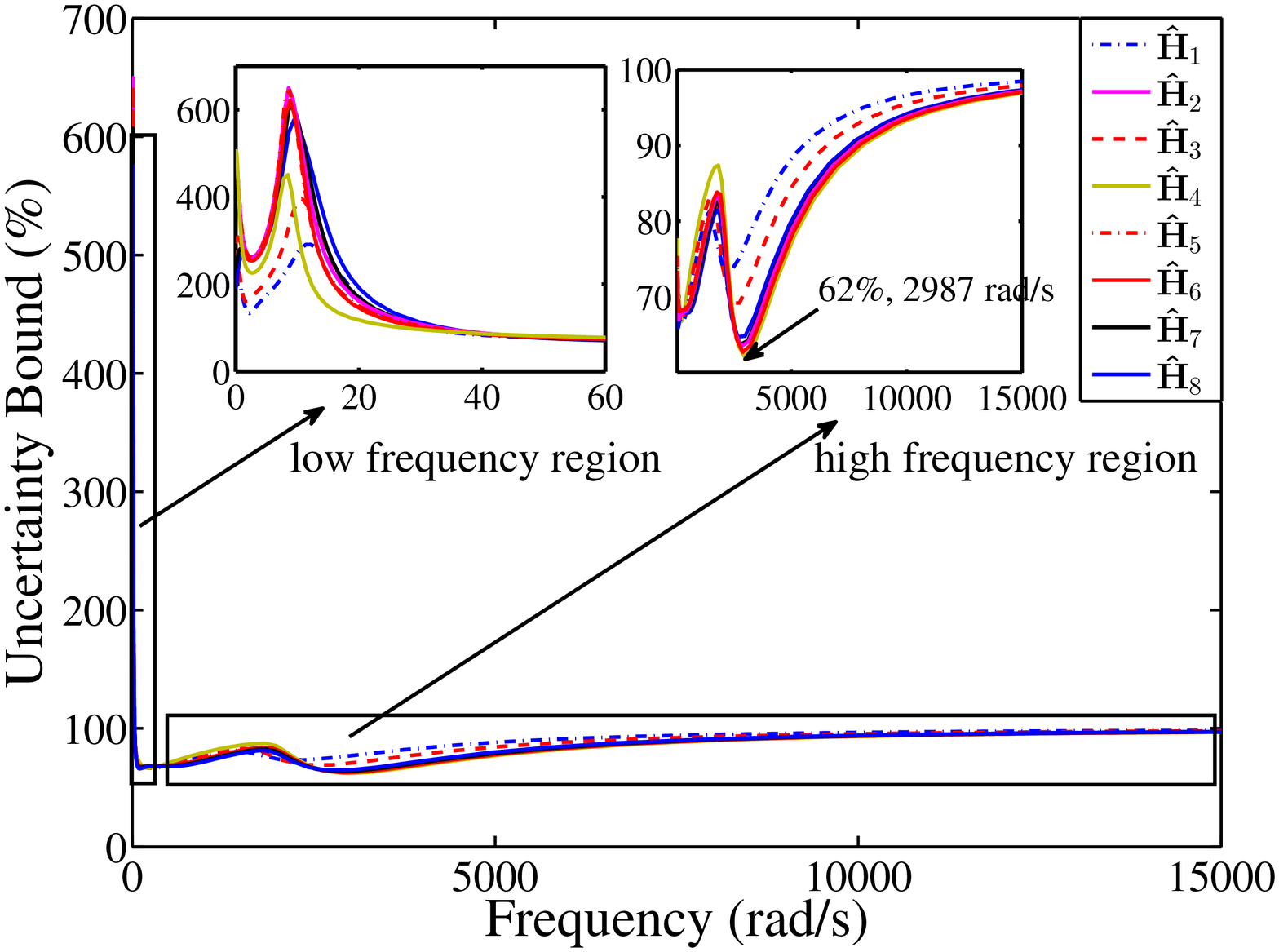}
\caption{Inverse input multiplicative uncertainty bound}
\label{invinmulunall3}
\end{center}
\end{figure}
\subsection{Performance and decoupling analysis}\label{PerAna}
The performance of the controller   is evaluated by analyzing the singular value plots of  output sensitivity function, $\mathbf{S_o}_f(s)=\mathbf{(I-\tilde{P}}_f(s)\mathbf{K)}^{-1}$ shown in Fig. \ref{fig:SoALL1}. 
The controller rejects any disturbance with a frequency less than 197~rad/s acting on any controllable output channels of the NAV as  Fig. \ref{fig:SoALL1} indicates  worst 0~dB crossing frequency (from below) of  $\overline{\sigma}(\mathbf{S_o}_f(j\omega))$ at the controllable output channel is 197~rad/s. The worst magnitudes of $\overline{\sigma}(\mathbf{S_o}_f(0))$ associated with $\bar{r}$, $\bar{\phi}$, and $\bar{\theta}$ output channels are -7.8~dB, -44.6~dB, and -59.8~dB, respectively. Hence, the worst maximum steady-state error for a step input at the controllable output channels are 40.7~$\%$, 0.6~$\%$, and 0.01~$\%$, respectively. The detailed performance analysis including  the input sensitivity,  $\mathbf{S_I}_f(s)$ and  $\mathbf{KS_o}_f(s)$ are given in the supporting material. 
The value of chosen desired eigenvector elements associated with the coupled spiral mode of  CL plant  of $\mathbf{\tilde{P}}_{cp}(s)$ are $\bar{u}$=-0.0065, $\bar{w}$=0.0905, and $\bar{r}$=-0.0002. For the coupled Dutch roll mode, the value of these elements are $\bar{p}$=-0.0159$\pm$0.0423$j$ and $\bar{\phi}$=-0.0010$\pm$0.0039$j$. Likewise, the chosen desired eigenvector element of coupled short period mode is $\bar{r}$=-0.0003$\pm$0.0001$j$. The values of these elements suggest that they are within the bounds given in Table~\ref{tab:dec1}.
 \begin{figure}[h!]
\begin{center}
\includegraphics[width=2.75in, height=1.75in]{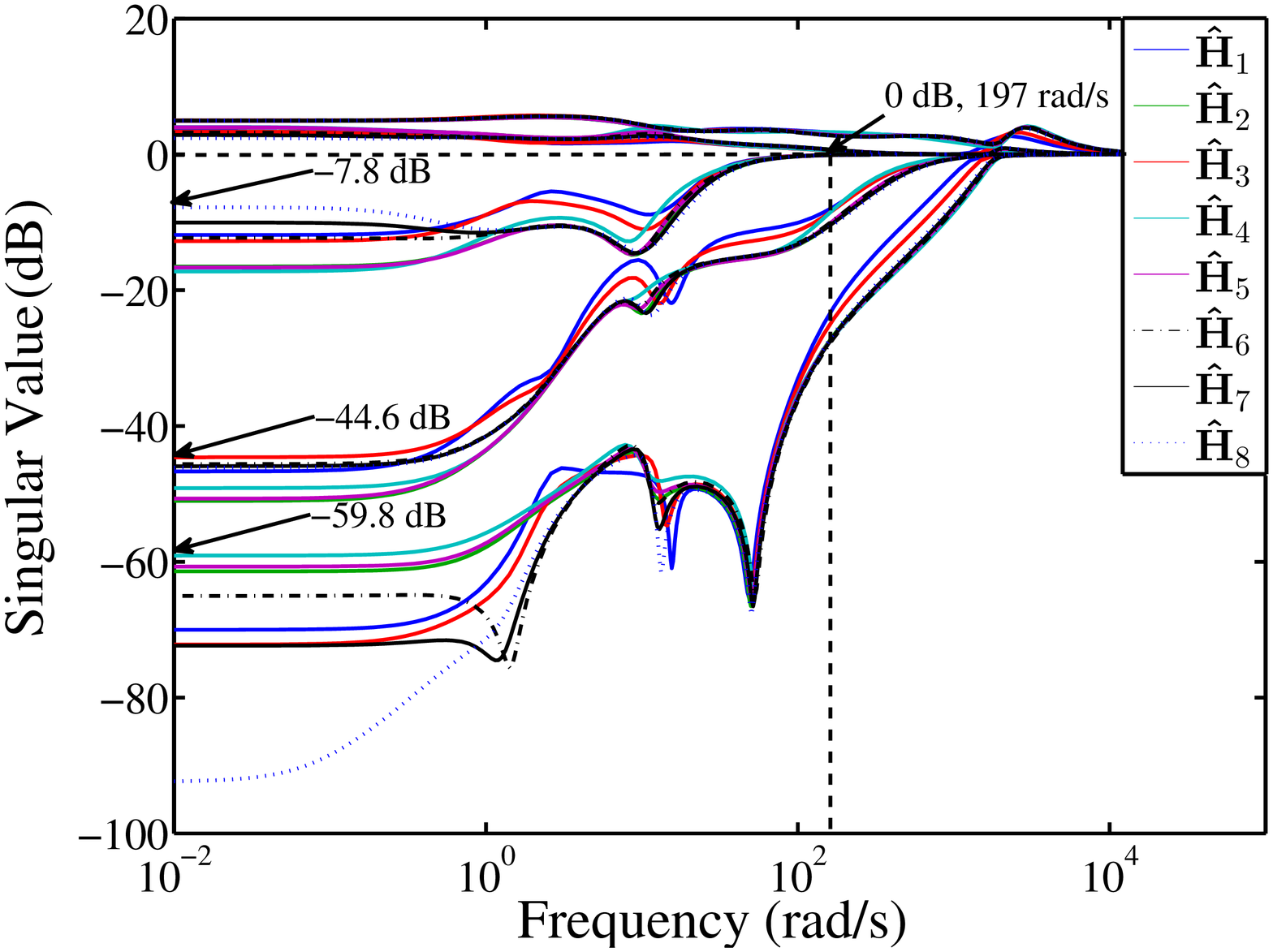}
\caption{Singular value plots of the output sensitivity functions}
\label{fig:SoALL1}
\end{center}
\end{figure}
 \subsection{Six-Degree-of-Freedom Simulations}
To assess the performance of the controller,  five cases of six-degree-of-freedom simulations are accomplished. In first and second cases of simulations,  the CL nonlinear and linear plants of the NAV are  forced to track a doublet $\theta$ and $\phi$ reference signals separately.  The block diagram of  the   system structure for the simulation of CL nonlinear dynamics of the NAV is given in supporting material. The  magnitude and  pulse width of the doublet $\theta$ reference signal are ${\pm0.0873}$~rad (${\pm5}$ deg) and  2~s, respectively.
 Similarly,  the doublet $\phi$ reference signal has a magnitude of ${\pm0.0349}$~rad (${\pm2}$~deg) and a pulse width of 2~s. The output time responses of  the  linear CL plant  of $\mathbf{\tilde{P}}_{cp}(s)$  offsetted with the trim values are shown in  Figs. \ref{fig:theF11}-\ref{fig:rcoF11p} along with  the output time responses of  the corresponding nonlinear  CL plant. The output time responses of  all the   CL plants are given in the supporting material. After analyzing all these plots, the following conclusions are made.
\begin{enumerate}
\item  The CL nonlinear plants  are stable as their  time responses of output variables are not diverging in finite time. The time responses of the CL nonlinear and linear  (with offset) plants are almost identical, which establishes that the linear plant  accurately captures the behavior of the nonlinear plant around the trim point. Furthermore, the CL plants are tracking their reference signals as the tracking variables are  within the specified error band.
  \item While tracking the doublet $\theta$ reference signal, the change in yaw rate responses  is minimal when compared to the roll rate and pitch rate responses. This indicates that the coupled short period mode of all the CL plants are decoupled from yaw rate response which is one of the design requirement.
  \item Compared to the change in yaw rate responses, a substantial change in roll rate responses are visible when all the CL plants track the  doublet $\theta$ command. This large change is mainly due to the variation of counter torque due to the change in $\delta_T$ and thereby thrust. The coupling induced by counter torque affects the coupled roll modes of all the CL plants and thereby the roll rates.  The controller is not designed to decouple the coupled roll modes from roll rate responses, which result in the excitation of the roll rates when the NAV is forced to track a doublet $\theta$ command. However, the converging (converging to the trim value) of roll rate responses establishes that the coupled roll modes of  all the CL plants are stable. The roll angle responses  indicate the variations in the roll angle responses are minimal. This is because the effect of  counter torque on the  coupled spiral mode is minimal as   $u$ and $w$ are decoupled from the coupled spiral mode. Furthermore, the time responses indicate that all the CL plants have similar desired decoupling characteristics and thus, the controller  accomplishes one of the design objectives.
\end{enumerate}
In the third case and the fourth case, the effectiveness of the controller in tolerating the parametric and unmodeled dynamic uncertainties is evaluated. To evaluate the effect of parametric uncertainty, 40~$\%$  of uncertainty is induced in the static, dynamic, and control   derivatives of the NAV  by   multiplying these derivatives with the signal (with the frequency of 25.13~rad/s) shown in Fig. \ref{untysig} while nonlinear CL plant of $\mathbf{\tilde{P}}_{cp}(s)$ tracks a doublet $\theta$ reference signal. 
Likewise, to evaluate the effect of unmodeled dynamic uncertainty,  an unmodeled dynamic uncertainty is modeled by adding the frequency-dependent weight, $G(s)=\frac{3s + 923.9}{s + 9239}$, to the closed-loop (as shown in the Fig. 8 of the supporting material).   This weight induces  50~$\%$ relative uncertainty up to 60 rad/s to the respective  output channel.   The relative uncertainty  rising thereafter induces 100$\%$ relative uncertainty at 3250 rad/s.
Corresponding  simulation results  indicate that the CL nonlinear plant  is stable, as the output responses are not diverging as shown in Figs. \ref{fig:theF11}-\ref{fig:rcoF11p}. Additionally, the CL nonlinear plant  is tracking its reference signal. The effectiveness of the controller in attenuating the disturbance is tested in the fifth case of simulation. 
Moreover,  the hardware-in-loop simulations (HILS) of the CL nonlinear plants of the NAV are also accomplished. The details of both these simulation results are given in supporting material. The video of HILS can be found in \href{https://youtu.be/UCscRGwaNqI}{https://youtu.be/UCscRGwaNqI}.
\begin{figure}[!h]
\centering
\subfigure[Uncertainty signal  \label{untysig}]{\includegraphics[width=1.72in, height=1.22in]{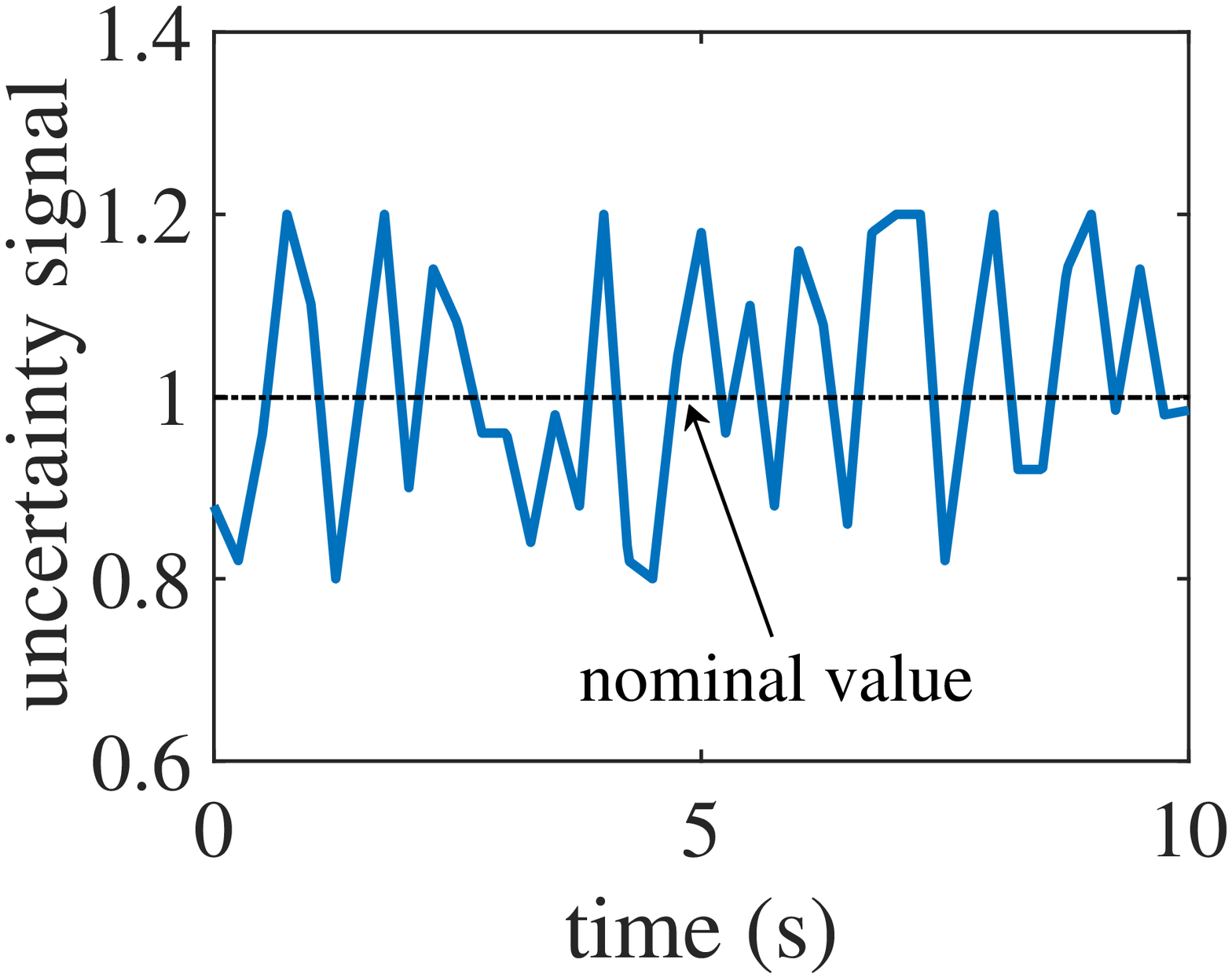}}
\centering
\subfigure[Pitch angle response  \label{fig:theF11}]{\includegraphics[width=1.72in, height=1.22in]{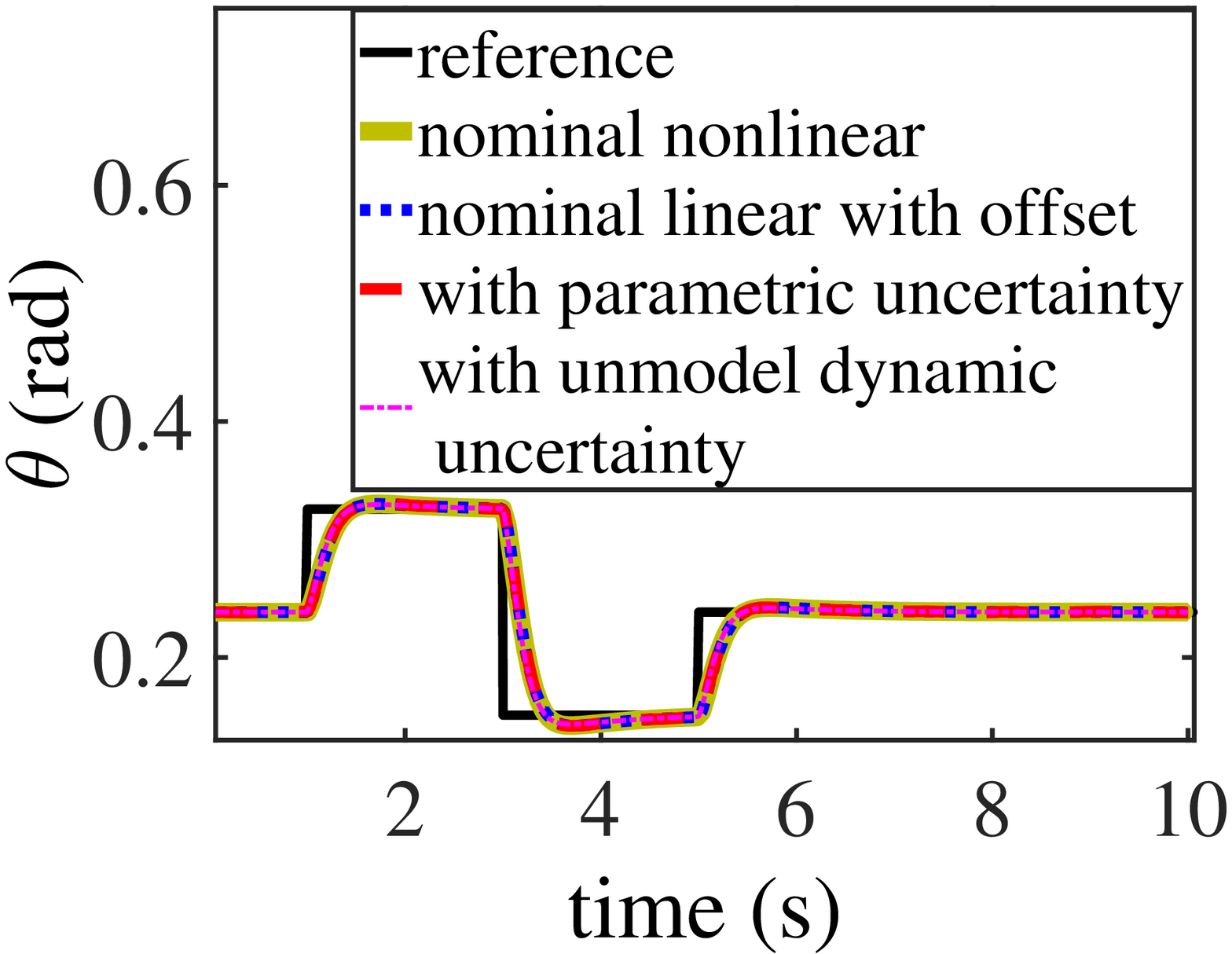}}
\subfigure[Roll angle response  \label{fig:ph11}]{\includegraphics[width=1.72in, height=1.22in]{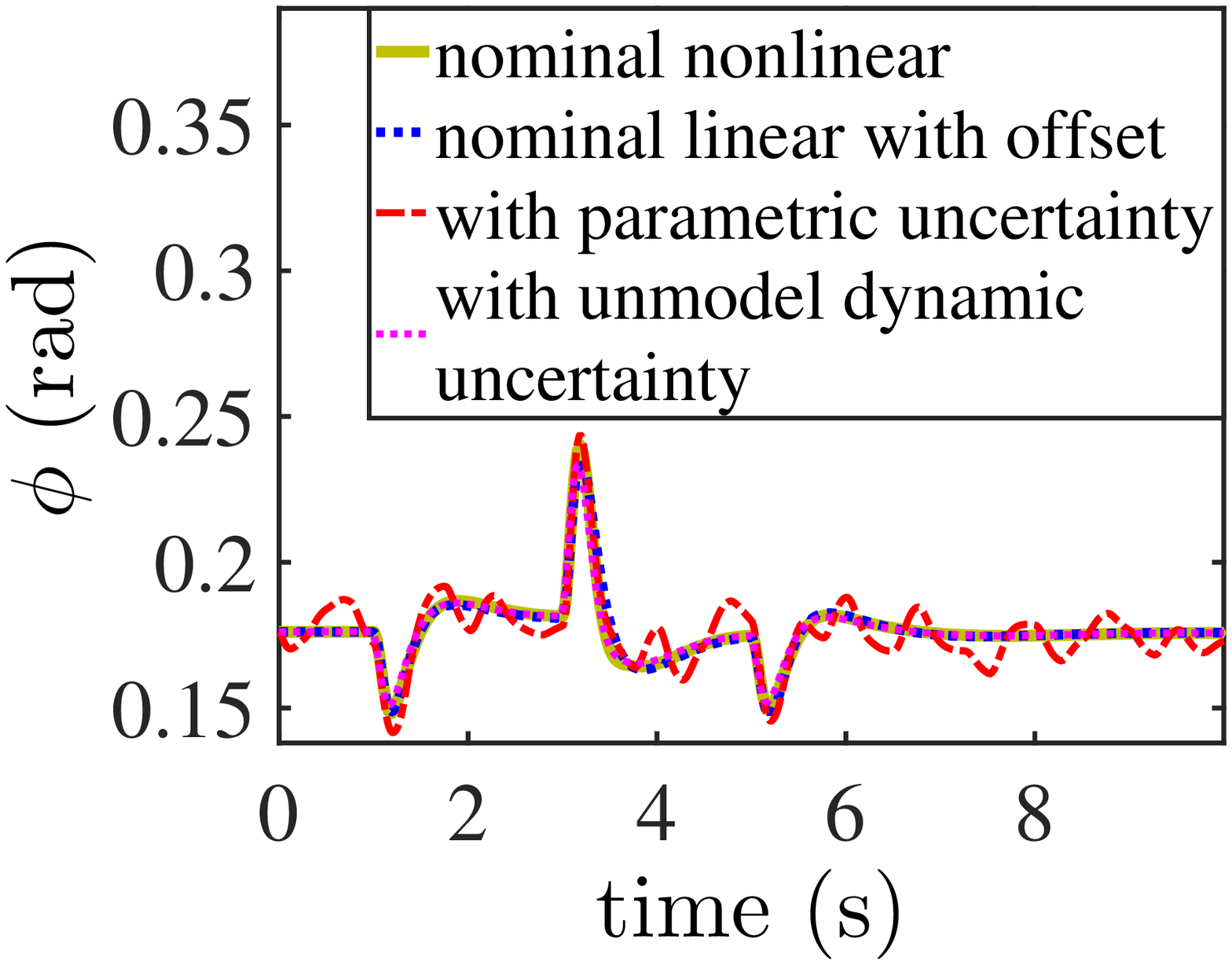}}
\subfigure[Roll rate response  \label{fig:pcoF11}]{\includegraphics[width=1.72in, height=1.22in]{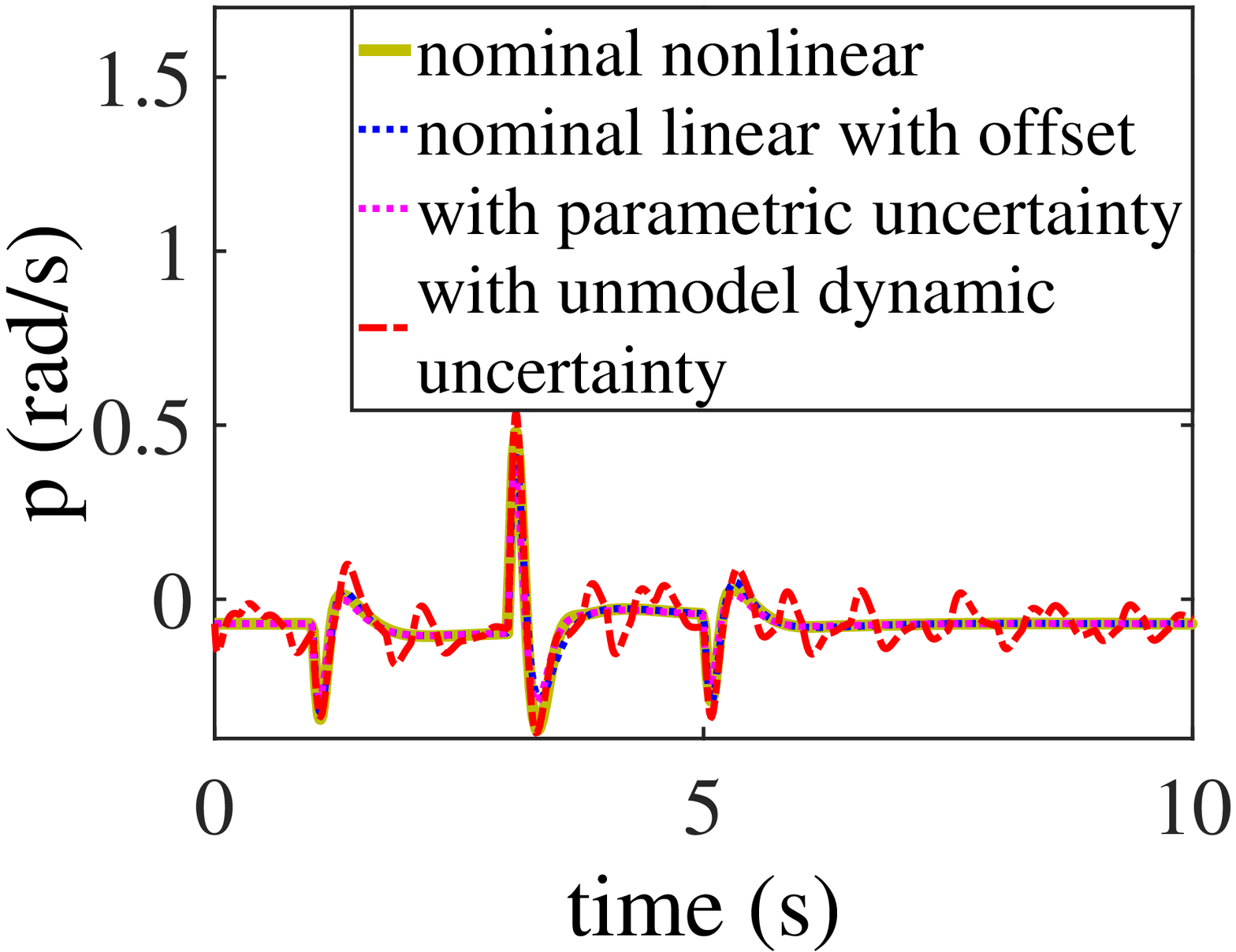}}
\subfigure[Pitch rate response  \label{fig:qcoF11}]{\includegraphics[width=1.72in, height=1.22in]{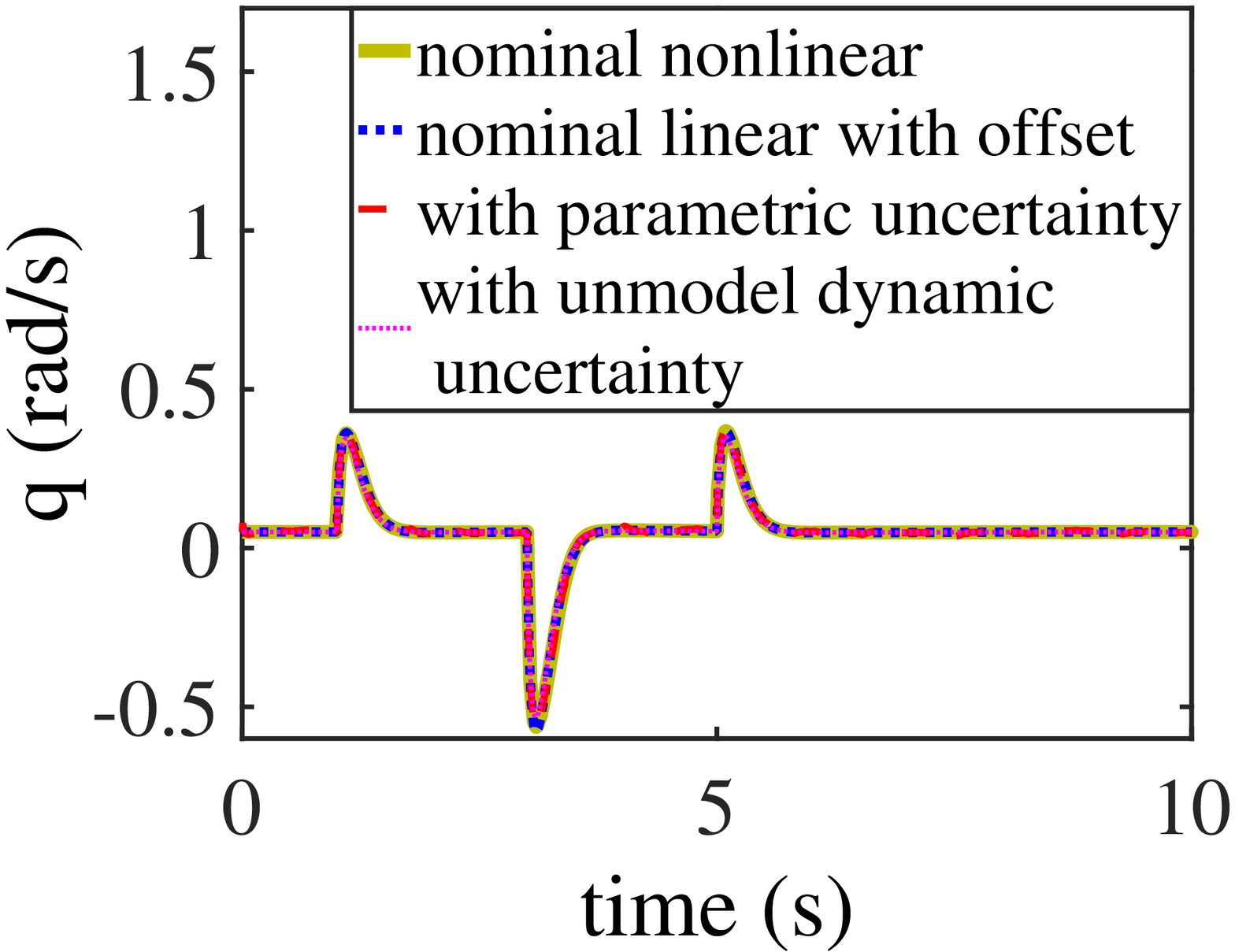}}
\subfigure[Yaw rate response  \label{fig:rcoF11}]
{\includegraphics[width=1.72in, height=1.22in]{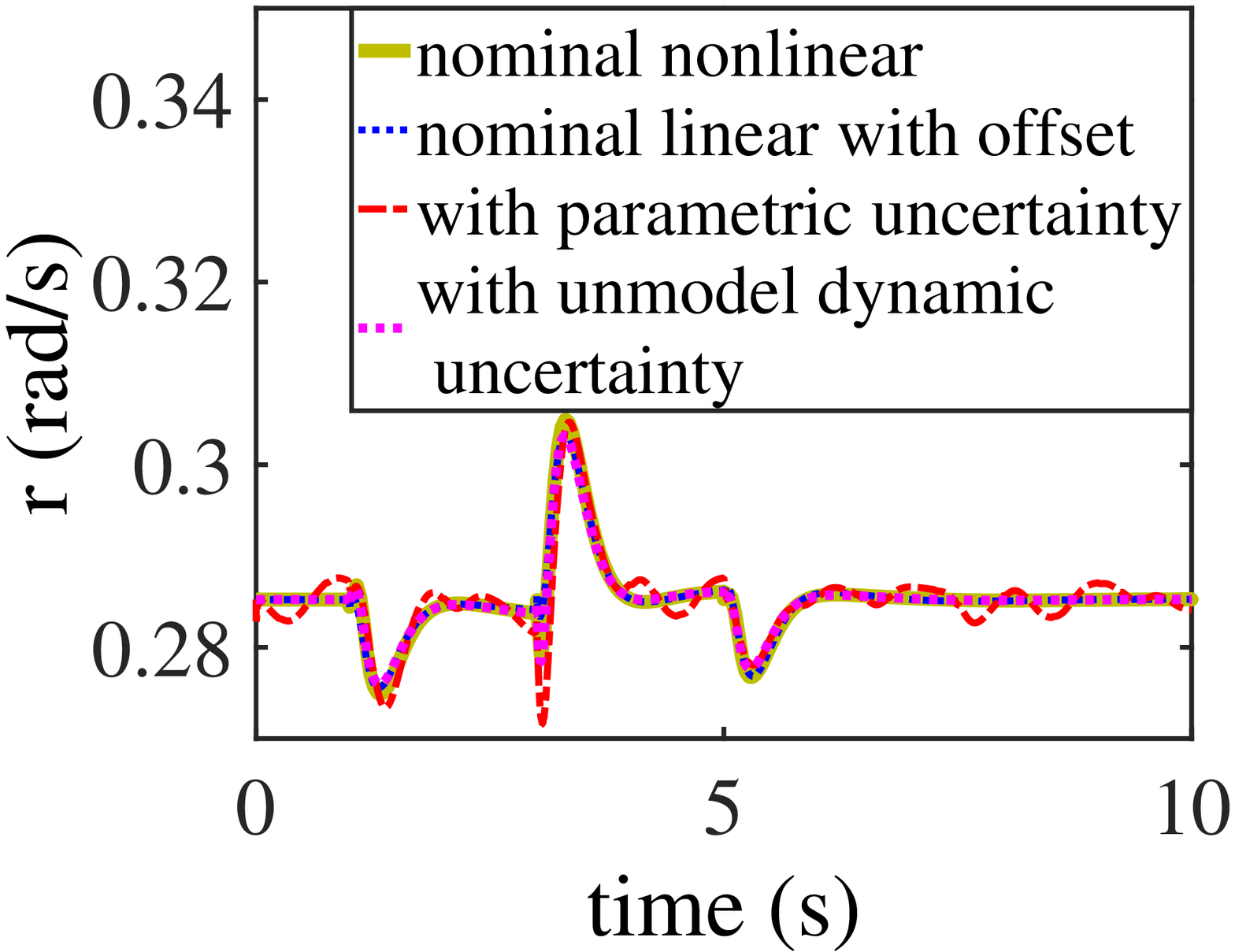}}
 \caption{$\theta$ tracking performance of  CL  plant  of $\mathbf{\tilde{P}}_{cp}(s)$}
\end{figure} 
\begin{figure}[H]
\centering
\subfigure[Pitch angle response  \label{fig:theF11p}]{\includegraphics[width=1.71in, height=1.16in]{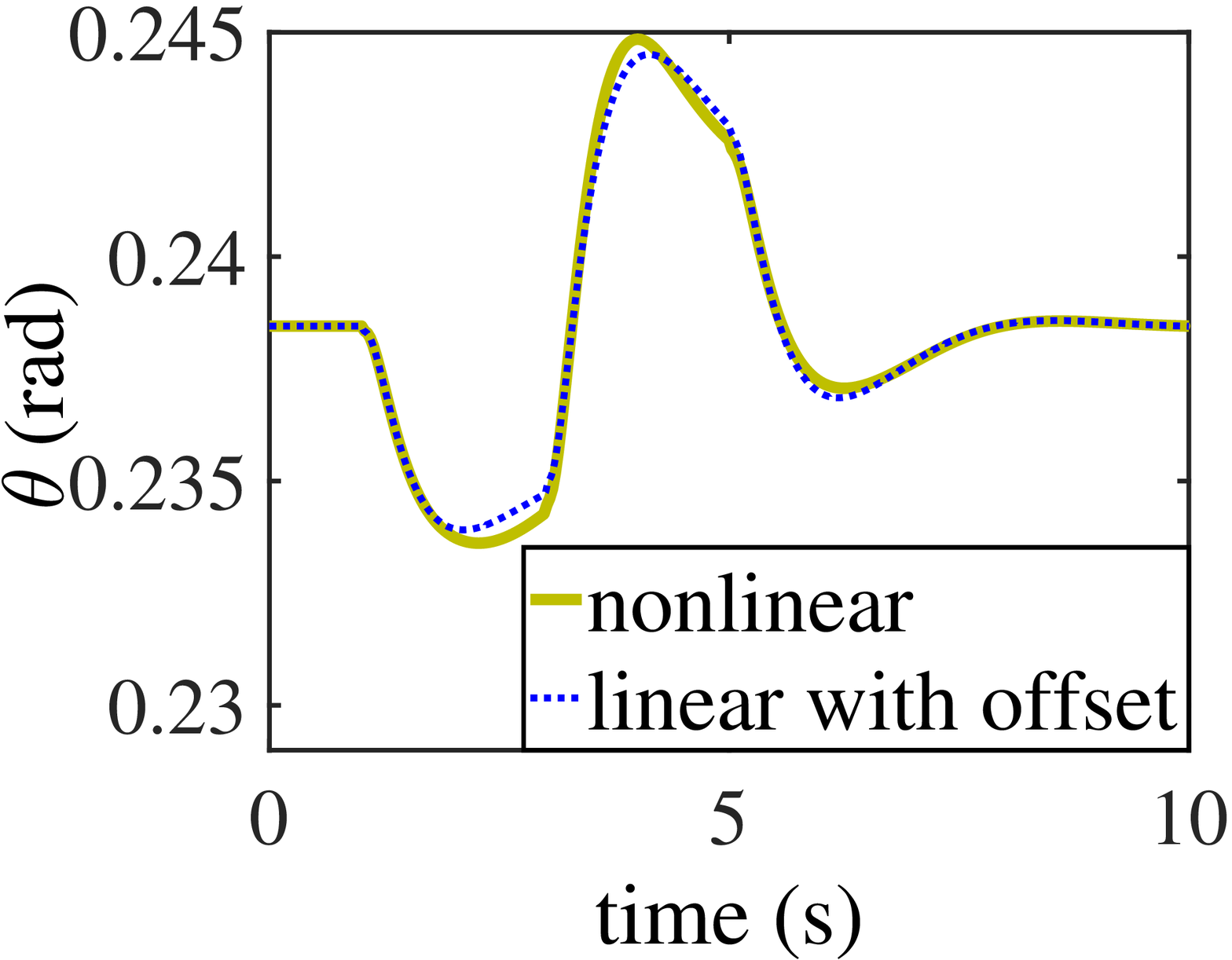}}
\subfigure[Roll angle response  \label{fig:ph11p}]{\includegraphics[width=1.71in, height=1.16in]{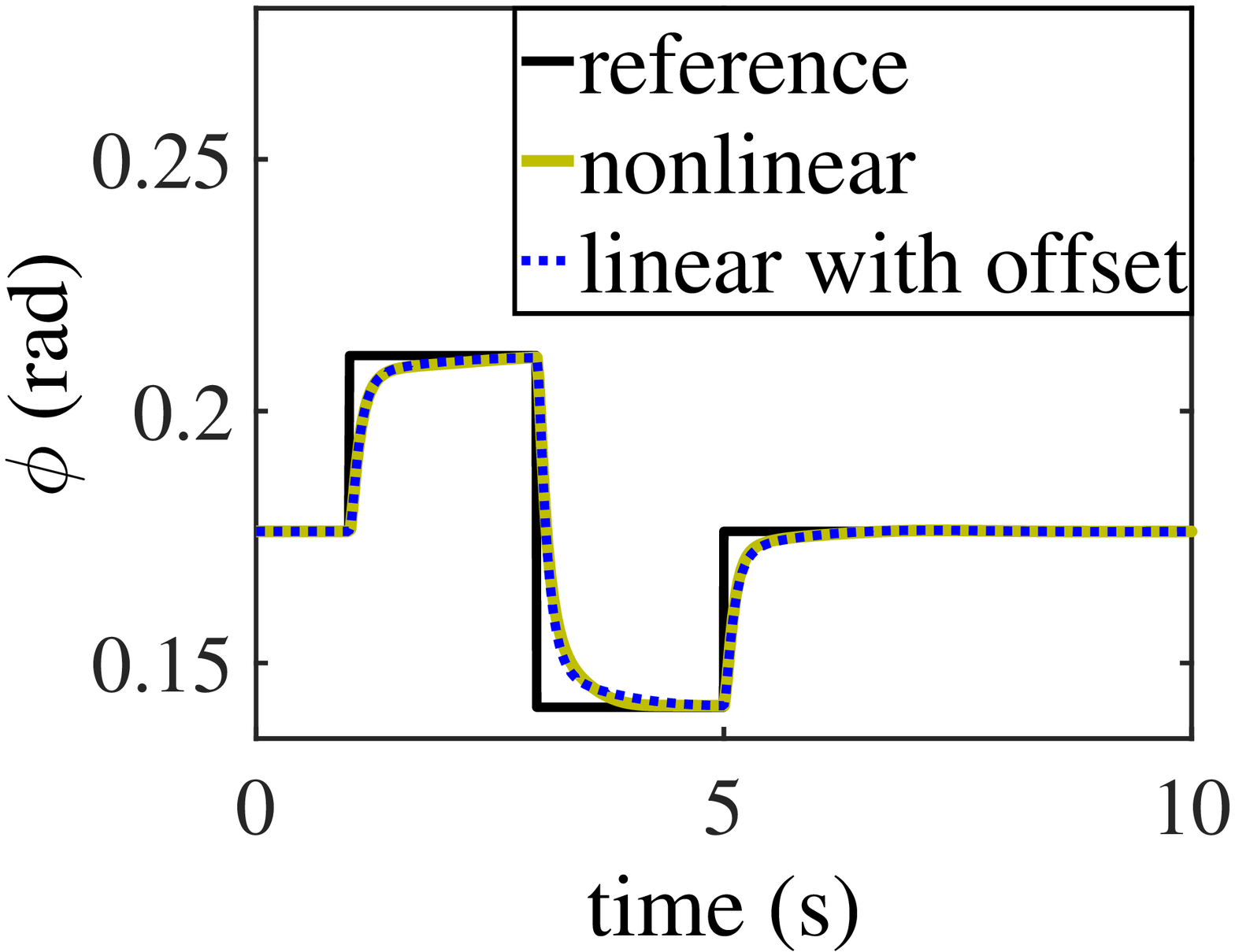}}
\subfigure[Roll rate response  \label{fig:pcoF11p}]{\includegraphics[width=1.72in, height=1.16in]{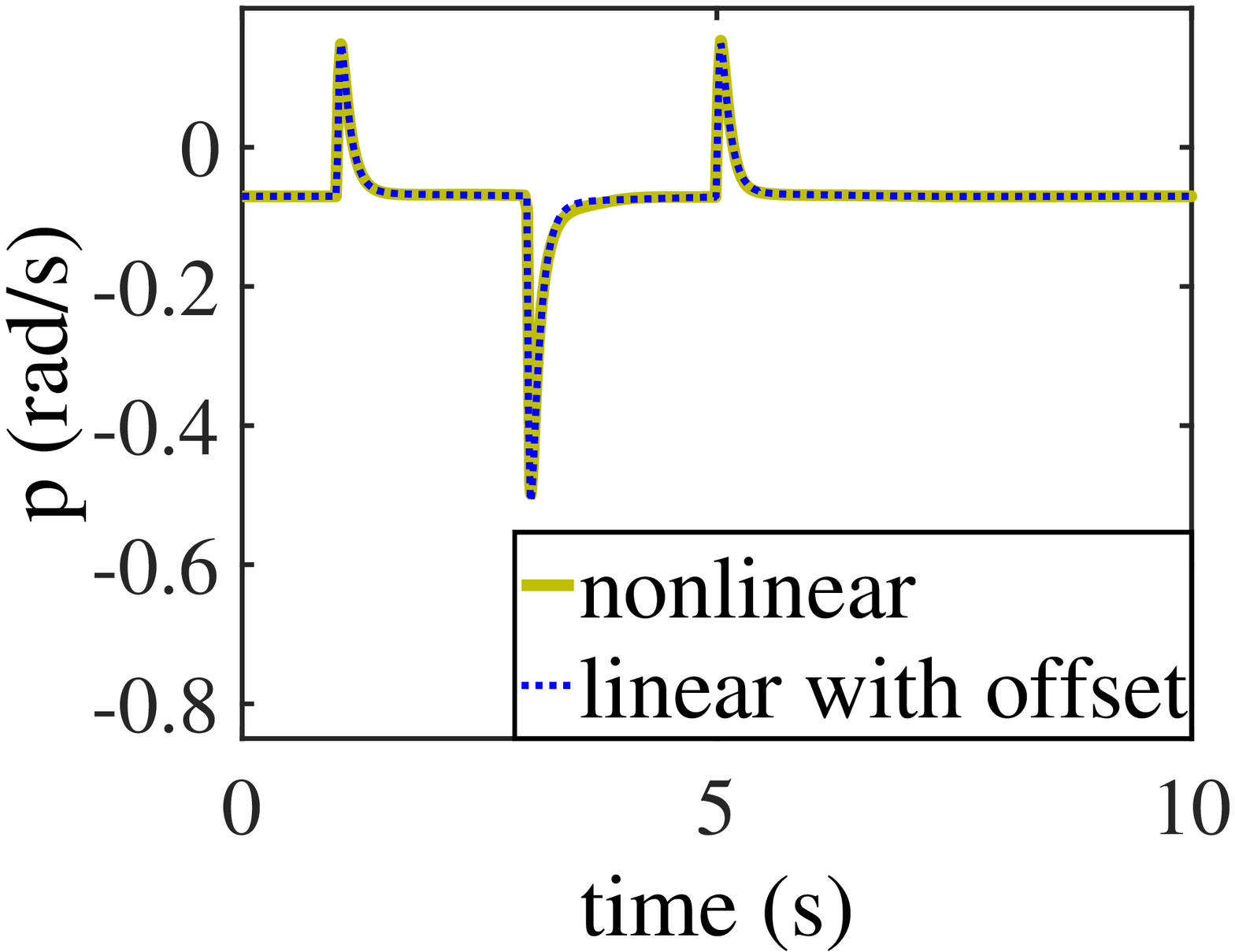}}
\subfigure[Pitch rate response  \label{fig:qcoF11p}]{\includegraphics[width=1.72in, height=1.16in]{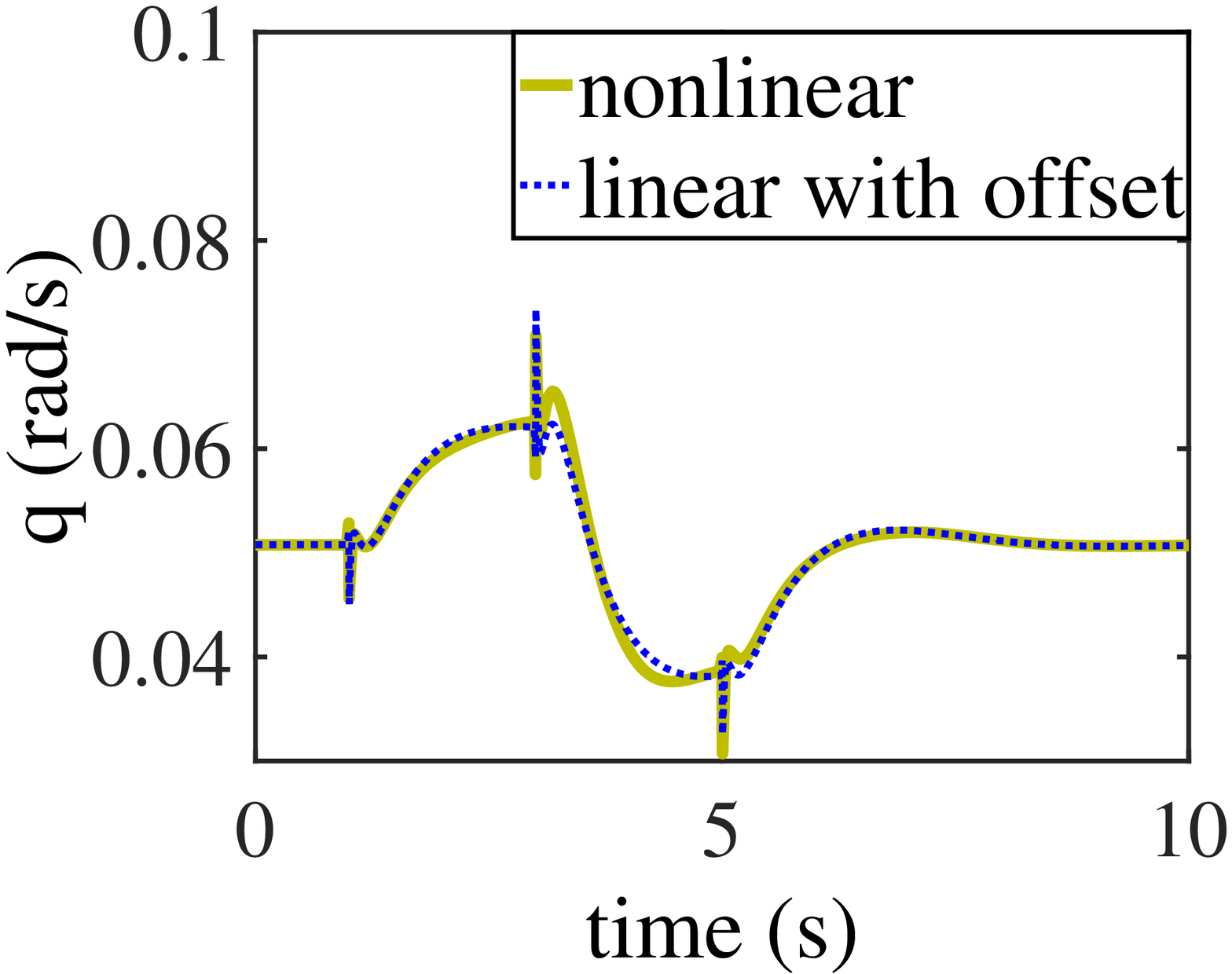}}
\subfigure[Yaw rate response  \label{fig:rcoF11p}]
{\includegraphics[width=1.7in, height=1.16in]{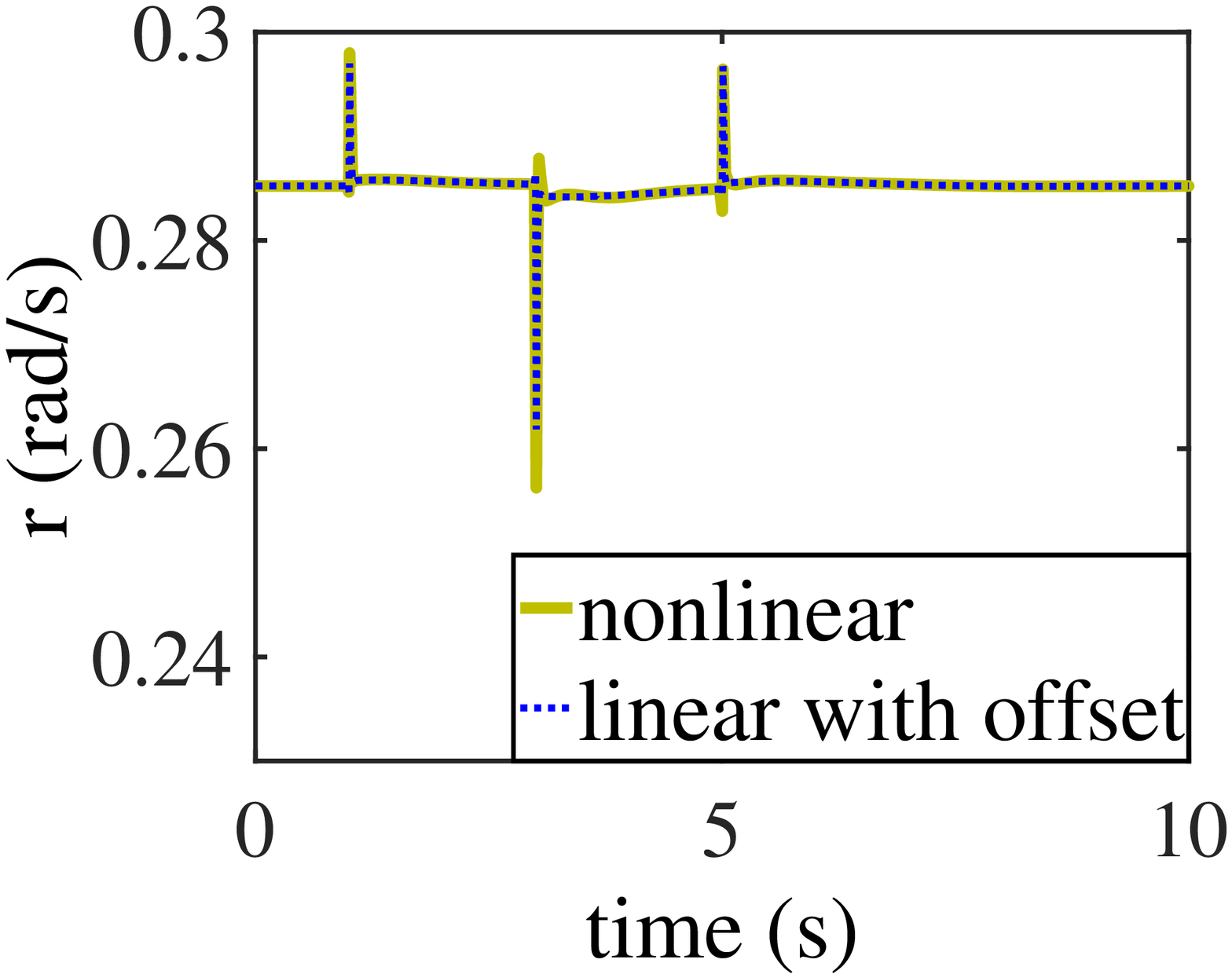}}
\caption{$\phi$ tracking performance  of CL   plant of $\mathbf{\tilde{P}}_{cp}(s)$}
\end{figure}

\section{Conclusions}\label{concu}
A new tractable method based on the central plant is developed to synthesize a RSSD  output feedback controller  for a finite set of unstable MIMO adversely
coupled plants of a NAV having  resource-constrained autopilot hardware. To this end, the method developed to identify the central plant  is tractable and is suitable for stable/unstable plants with the varying/same number of unstable poles. The sufficient conditions for the existence of the RSSD  output feedback controller developed are easily testable for a given set of plants. Further, the tractability of the method developed using these conditions to solve the RSSD problem is successfully demonstrated by generating a feasible RSSD output feedback controller for the eight unstable plants of the fixed-wing NAV. All the CL plants with this controller demonstrate that they satisfy desired design specifications as indicated by the stability, performance,
and decoupling analyzes and the six-degree-of-freedom simulation results. Additionally, the hardware-in-the-loop simulation results show that there are no  implementation issues associated with the controller and the compensators.

\vspace{-1cm}
\begin{IEEEbiographynophoto}{Jinraj V Pushpangathan}
received his Ph.D. degree in aerospace engineering from Indian Institute of Science (IISc), India in 2018. Currently, he is  working as research fellow in Aerospace Department of IISc.   His research interests are robust and optimal control theory, guidance and control of unmanned systems, and flight dynamics and control.
\end{IEEEbiographynophoto}
\vspace{-1.2cm}
\begin{IEEEbiographynophoto}{K. Harikumar}
is currently an Assistant Professor in International Institute of Information Technology, Hyderabad, India. He  received the Ph.D. degree in  aerospace engineering from the Indian Institute of Science (IISc), India in 2015.  His research interests are applications of control theory to unmanned systems and flight dynamics.
\end{IEEEbiographynophoto}
\vspace{-1.2cm}
\begin{IEEEbiographynophoto}{Suresh Sundaram (SM'08)}
  received the Ph.D. degrees in aerospace engineering from the Indian Institute of Science (IISc),  India, in  2005, where he is currently an Associate Professor.
 From 2010-2018, he was an Associate Professor with the School of Computer Science and Engineering, Nanyang Technological University.
  His current research interests include machine learning,  optimization, and computer vision.
\end{IEEEbiographynophoto}
\vspace{-1.2cm}
\begin{IEEEbiographynophoto}{Narasimhan Sundararajan (LF'11)}
	received  the Ph.D. degree in electrical
	engineering from the University of Illinois at
	Urbana-Champaign, Urbana, IL, USA, in 1971.
	From 1971 to 1991, he was with the Vikram
	Sarabhai Space Centre, Indian Space Research Organization, Trivandrum, India. From 1991, he was a Professor (Retd.) with the School of Electrical and Electronic Engineering, Nanyang
	Technological University (NTU), Singapore.  He was a Senior Research Fellow with the School of Computer Engineering, NTU. His current research interests include  spiking neural networks, neuro-fuzzy systems, and optimization with swarm intelligence.
\end{IEEEbiographynophoto}

\end{document}